\documentclass[12pt, draftclsnofoot, onecolumn]{IEEEtran}
 \usepackage{amsmath,amssymb}
 \usepackage{subfigure}
 \usepackage{graphicx,graphics,color,psfrag}
 \usepackage{cite,balance}
 \usepackage{caption}
 \allowdisplaybreaks
 \usepackage{algorithm}
 \usepackage{accents}
 \usepackage{amsthm}
 \usepackage{bm}
 \usepackage{algorithmic}
 \usepackage[english]{babel}
 \usepackage{multirow}
 \usepackage{enumerate}
 \usepackage{cases}
 \usepackage{stfloats}
 \usepackage{mathrsfs}
 \usepackage{dsfont}
 \usepackage{color,soul}
 \usepackage{amsfonts}
 \usepackage{cite,graphicx,amsmath,amssymb}
 \usepackage{subfigure}
 \usepackage{eufrak}
 \usepackage{fancyhdr}
 \usepackage{hhline}
 \usepackage{graphicx,graphics}
 \usepackage{array,color}
 \usepackage{amsmath}
 \usepackage{mathrsfs }
\usepackage[mathscr]{eucal}

\newtheorem{theorem}{Theorem}

\newtheorem{definition}{Definition}

\newtheorem{lemma}{Lemma}

\newtheorem{assumption}{Assumption}

\newtheorem{remark}{\bf Remark}

\def\E{\mathsf{E}}

\def\phi{\varphi}

\def\l{\left}
\def\r{\right}
\def\({\left(}
\def\){\right)}

\def\b0{{\mathbf{0}}}

\newcommand{\nn}{\nonumber}

\begin{document}

%\addtolength{\textfloatsep}{-30pt}
\addtolength{\abovecaptionskip}{-2mm}
%
%\textheight=24cm

%%%%%%%%%%%%%
% TITLE
\title{\huge  Stochastic Control of Computation Offloading to a Helper with a Dynamically Loaded CPU}
\author{Yunzheng Tao, Changsheng You, Ping Zhang, and Kaibin Huang
  \thanks{\noindent Y. Tao and P. Zhang are with the State Key Laboratory of Networking and Switching Technology, Beijing University of Posts and Telecommunications, Beijing 100876, China. C. You and K. Huang are with the Department of Electrical and Electronic Engineering, The University of Hong Kong, Hong Kong. Corresponding author: K. Huang (Email: huangkb@eee.hku.hk).}}
\maketitle

\begin{abstract}

Due to densification of wireless networks, there exist abundance of idling  computation resources at (network) edge devices (e.g., access points and handheld computers). These resources can be scavenged by offloading heavy computation tasks from  small IoT devices  (e.g., sensors and wearable computing devices) in proximity, thereby overcoming their limitations and lengthening their battery lives. However, unlike dedicated servers, the \emph{spare} resources offered by edge helpers are \emph{random} and \emph{intermittent}. Thus, it is essential for a user (IoT device) to intelligently control the amounts of data for offloading and local computing so as to ensure a computation task can be finished in time consuming  minimum energy. In this paper, we design energy-efficient control policies in a computation offloading system with a random channel and a helper with a dynamically loaded CPU (due to the primary service). Specifically, the policy adopted by the helper aims at determining the sizes of offloaded and locally-computed data for a given task in different slots such that the total energy consumption for transmission and local CPU is minimized under a task-deadline constraint. As the result, the polices endow an offloading user robustness against channel-and-helper randomness besides balancing  offloading and local computing.  By modeling the channel and helper-CPU as Markov chains, the problem of offloading control is converted into a Markov-decision process. Though \emph{dynamic programming} (DP) for numerically solving the problem does not yield the optimal policies in closed form, we leverage the procedure to quantify the optimal policy structure and apply the result to design optimal or sub-optimal policies. For different cases ranging from zero to large buffers, the low-complexity of the policies overcomes the ``curse-of-dimensionality'' in DP arising from joint consideration of channel, helper CPU and buffer states.
%Specifically, the optimal policy for the case of no helper buffer is derived in closed form. Based on a similar structure and approximating the  objective function, sub-optimal polices are designed for the more challenging cases of large and small helper buffer. Simulation demonstrates their close-to-optimal performance.
\end{abstract}

\section{Introduction}
Mobile edge computing  has emerged as a promising  technology for realizing the vision of \emph{Internet of Things} (IoT) by reducing computation latency and energy consumption of IoT devices. These advantages are achieved by  \emph{mobile-edge computation offloading} (MECO), a primary operation of edge computing, that offloads  complex  computation from small IoT devices  to nearby edge  helpers (or servers) such as \emph{access points} (APs) and  smartphones \cite{patel2014mobile,mao2017survey}. Among edge helpers, while some are \emph{dedicated} servers, others provide  edge computing \emph{opportunistically} as a secondary service without affecting  their primary functions. For example, the primary function of a base station is to support  radio access and networking and that of a smartphone is   personal computing. In other words, only when not performing the primary function, an opportunistic helper offers its idling CPU for the intermittent use by  an offloading user. Then from the perspective of an offloading user, an opportunistic helper appears to be one with a dynamically loaded CPU. The randomness in both the available computing resources and wireless channel requires a user to adapt the offloading process to the states of the helper and channel. Specifically, by observing the states, the user should control  1) the size of input-data for a particular task to be computed  in the current slot and 2) its partitioning for computing ``\emph{remotely}'' and ``\emph{locally}'', such that the task can be completed in time with the minimum energy consumption.  In this paper, we study this complex problem of optimal MECO control using stochastic optimization theory.

\subsection{Prior Works}
\subsubsection{Mobile-edge computation offloading}
 MECO is a key enabler for a wide range of  computation-intensive and latency-sensitive mobile applications, such as augmented/virtual reality, high-definition video streaming and online gaming. This drives growing research interests in both the academia and industry to design efficient computation-offloading systems and techniques \cite{zhang2013energy,you2016energy,mahmoodi2016optimal,wang2016mobile,ko2017live}. For single-user MECO systems, the energy-efficient \emph{binary} offloading was investigated in \cite{zhang2013energy} where the CPU-cycle frequency and offloading rates were optimized for reducing the energy consumption of local computing and offloading, respectively. This work was extended in \cite{you2016energy} by powering MECO with wireless energy. Subsequently,  \emph{partial offloading} allowing more flexible offloading than binary offloading  has been designed for enhancing the performance of energy savings and latency reduction using diverse techniques such as  adaptive program partitioning \cite{mahmoodi2016optimal, wang2016mobile} and live prefetching based on computation prediction  \cite{ko2017live}. For multiuser MECO systems, the radio-and-computational resource management has been intensively investigated \cite{you2017energy,chen2016efficient,lyu2017multiuser,yang2015multi}. Specifically, an optimization problem was formulated in \cite{you2017energy} to minimize the weighted sum user-energy consumption. The derived optimal centralized resource-allocation policy is shown to have a simple threshold-based structure. Moreover, the distributed energy-efficient resource-allocation for multiuser MECO systems was studied in \cite{chen2016efficient, lyu2017multiuser} by using game theory and decomposition techniques, respectively. In addition, an energy-efficient multiuser MECO scheduling policy was proposed in \cite{yang2015multi} by assigning different mobiles with diverse levels of priorities based on their latency requirements. Last, the latency performance of large-scale MECO networks was analyzed in \cite{ko2017wireless}  using stochastic geometry.

Recently, \emph{peer-to-peer} (P2P) computation offloading between mobile devices has emerged to be  an active research area \cite{song2014energy,li2014exploring,chen2017exploiting,ti2017computation,xu2016less,cao2017joint,you2017exploiting}. The technology can enhance   the computing capabilities of mobiles by sharing as well as equalizing  the uneven distribution of computation workloads and computation resources over heterogeneous devices. Specifically, an energy-efficient peer computation offloading framework was proposed in \cite{song2014energy}, where mobiles cooperatively compute tasks and share computation results to balance the energy consumption of mobile users. From the perspective of wireless transmission, P2P computation offloading can be implemented by \emph{device-to-device} (D2D) communications \cite{li2014exploring,chen2017exploiting,ti2017computation}. The interplay between user incentives and interdependent security risks in D2D offloading was investigated in \cite{xu2016less} by leveraging tools from game theory and epidemic theory. Last, a protocol for  helper-assisted computation offloading to a central cloud  was proposed in\cite{cao2017joint} and designed for maximizing user energy efficiency via the  energy-efficient joint computation-and-communication control.

% \hl{For simplicity, use the term "opportunistic helpers" in the paper.}

 All of the above prior works  assume dedicated  helpers or servers (see e.g., \cite{you2017energy,ti2017computation}). This, however, overlooks the fact that many helpers in practice are \emph{opportunistic} since their primary functions are not  edge computing. Recently, computation offloading to an opportunistic helper was studied in \cite{you2017exploiting} where the user exploits   \emph{non-causal} information on the helper-CPU state to
 control offloaded data sizes in different slots. Computation prediction for acquiring such information may not be accurate in practice and furthermore places an extra burden on the CPU. In this work, the offloading is controlled instead based on the distribution of the dynamic helper CPU, as such information can be easily inferred from historical data. The corresponding optimal control policy can be derived using stochastic optimization theory.

\subsubsection{Stochastic  computation offloading} The stochastic control of MECO has been extensively studied in the literature targeting a time-varying channel \cite{zhang2015collaborative,liu2016delay,jia2014heuristic,huang2012dynamic,mao2016dynamic,kwak2015dream,chen2017energy}. In the area of wireless communications, energy-efficient  data transmission over dynamic channels under a deadline constraint is a classic topic where the optimal policies can be typically computed using \emph{dynamic programming} (DP) assuming  the Markov-chain channel models (see e.g.,\cite{fu2006optimal}).  However, for stochastic control of  computation offloading, the policy design is  more complicated as it requires the joint control of  local computing and offloading (transmission), accounting for different types of  dynamics arising in channels, computation tasks and helper CPUs. The challenge has been tackled recently by a series of research \cite{zhang2015collaborative,liu2016delay,jia2014heuristic,huang2012dynamic,mao2016dynamic,kwak2015dream, chen2017energy}. In \cite{zhang2015collaborative}, considering the joint mobile-and-cloud task scheduling,  the problem of optimal offloading control is formulated as a shortest path problem and solving it reveals the ``one-climb" policy structure (i.e., the tasks should only migrate at most once between the mobile and the cloud). In the presence of random channel and computation-data arrivals, a stochastic offloading algorithm for controlling offloading decision and rate was designed in \cite{liu2016delay} using the approach of \emph{Markov decision process} (MDP). Accounting for different task-call graphs, a set of online task offloading algorithms were proposed in \cite{jia2014heuristic} to minimize the completion time for sequential or concurrent tasks. Another line of research assumes that only the \emph{instantaneous} information (e.g., random channel, data arrival and renewable energy) is available at the user and applies the  Lyapunov optimization techniques to design the control policies for  controlling computation offloading\cite{huang2012dynamic,mao2016dynamic}, integrated with the CPU-frequency control for local computing {\cite{kwak2015dream} or server on/off switching \cite{chen2017energy}.

Again, the prior works assume dedicated edge servers and do not consider opportunistic helpers that widely exist in practice. Accounting for the corresponding randomness in computation resources in addition to channel dynamics  complicates the stochastic offloading control.  This work thus aims to bridge this gap by investigating stochastic offloading control targeting both dynamic channel and dynamic helper CPU.

\subsection{Contributions}

To the best of the authors' knowledge, this paper represents the first attempt on investigating the optimal stochastic control of MECO to an \emph{opportunistic} helper. In particular, the work addresses two new design issues that are not addressed in the literature focusing on dedicated helpers.

\begin{itemize}

\item \emph{(Channel-and-helper dual opportunism)} The dual opportunism refers to exploiting both the channel temporal diversity and the  random computation resources at the helper for reducing transmission-energy and computing-latency, respectively. Note that  the intermittent  computation resources at the helper can only be utilized in \emph{real time} but not earlier or later. Then  the offloading control policy  should be designed to balance the offloading of sufficient  input-data prior to the availability of random helper CPU and delaying offloading so as to exploit channel temporal diversity. Thus the policy design is more complex than the conventional ones without helper randomness (see e.g., \cite{li2014exploring,chen2017exploiting,ti2017computation}).

\item \emph{(Curse of dimensionality)} The second issue  is the dramatic increasing in the complexity of computing offloading control policy due to a larger state space including both the channel, helper CPU and helper buffer, which is the effect of the well-known curse-of-dimensionality in stochastic optimization. The problem is even more complicated given an adjustable CPU frequency for local computing, as considered in the current work. In contrast, given non-causal information on the helper state, the problem of optimal offloading control in \cite{you2017exploiting} is a convex problem and much simpler to solve.
\end{itemize}

To address these issues,  we consider an MECO system where a user is assisted by an opportunistic helper  to compute a fixed size  of input-data before a deadline. To reduce energy consumption, the user allocates the data  over different slots for computation and furthermore partitions the per-slot data for local computing with an adjustable local-CPU frequency  and offloading to an opportunistic helper. Both the dynamic channel and helper-CPU states are modeled as Markov chains. The models allow the approach of transforming the considered offloading control into stochastic optimization problems. The main results are summarized as follows.

\begin{itemize}
  \item[1)] \emph{Opportunistic helper  without a buffer}: First, consider the case where the helper has no buffer for storing the offloaded data from the user. We formulate a stochastic  optimization problem for  minimizing  the expected user-energy consumption and derive the optimal offloading policy for controlling data partitioning (for local computing and offloading) in different slots. The result shows that, in each slot, both the locally-computed and offloaded data sizes are \emph{proportional} to the remaining un-computed data size with the scaling factors  depending on the instantaneous channel and helper-CPU states.

  \item[2)] \emph{Opportunistic helper  with a large buffer}: Solving the problem of stochastic offloading control for the case of a large helper buffer is mathematically intractable. To overcome the difficulty, we derive an \emph{approximated} function for the expected future user-energy consumption. Based on this, we propose an effective sub-optimal computation offloading policy, which is shown to be close-to-optimal in simulations. Interestingly, the resultant  policy has the similar  structure as that of zero-buffer counterpart but is much simpler to compute.

  \item[3] \emph{Opportunistic helper  with a small buffer}: Last, assume that the helper has a small buffer. The introduced buffer constraint renders the derivation for the optimal computation offloading policies intractable. To address this challenge, we propose a practical offloading control  scheme that switches between two  candidate policies  for the earlier  cases of zero buffer and  large buffer. The switching threshold can be  computed by a simple  bisection-search. Despite its simplicity, the proposed scheme achieves close-to-optimality in terms of energy savings as demonstrated by simulation results.

\end{itemize}

The rest of this paper is organized as follows. The system models and problem formulation are presented in Section~\ref{sec: Mathematical Model} and Section~\ref{Sec:Problem Formulation}, respectively. The optimal and suboptimal policies for controlling computation offloading are designed for three cases, namely zero, large and small buffers, in Sections~\ref{Sec:ContWoBuf}, \ref{Sec:LargeBuffer}, and \ref{sec:small_bf}, respectively. Section~\ref{sec: simulation results} provides simulation results, followed by concluding remarks in Section~\ref{sec: Concluding Remarks}.

\section{System Model}\label{sec: Mathematical Model}
Consider an MECO system shown in Fig.~\ref{Fig:sys}, consisting of one user and one opportunistic helper both equipped with a single-antenna. The user is required to compute $D$-bit input-data with $1$-bit data-size unit before a deadline $T$, which is divided into $K$ slots, each with a duration of $t_0=T/K$. Controlled by the helper, the user allocates data for computing in each slot and further partitions the data for simultaneous local computing and computation offloading to the helper. Models are described as follows.

\begin{figure}[t]
  \centering
  \includegraphics[height=5.3cm]{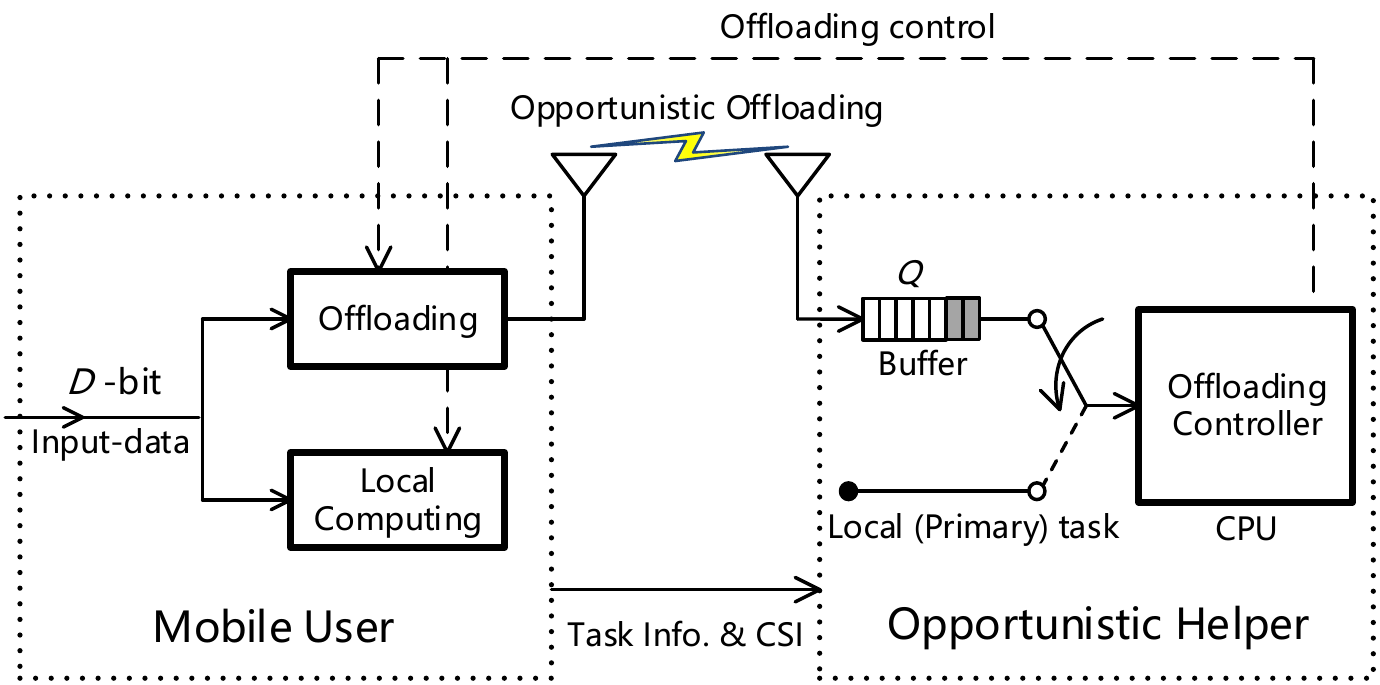}\\
  \caption{Computation offloading system with an opportunistic helper.}\label{Fig:sys}
\end{figure}

\subsection{helper CPU Model}\label{Sec:CPU}
The opportunistic helper is assumed to operate at a constant frequency. Due to intermittent primary tasks, the helper-CPU states can be modeled by a random process denoted by $\{C_1,\cdots, C_K\}$. In each slot $k$, the helper-CPU state $c_k\in\mathcal{C}\overset{\triangle}{=}\{0,1\}$, where $c_k=0$ and $c_k=1$ denote the busy and idle states, respectively.
\begin{assumption}[helper CPU Dynamics]\label{As:CPU}\emph{As shown in Fig.~\ref{Fig:helper-CPU}, the process of the random helper-CPU states, $\{C_k\}$ for $k=1, \cdots, K$, is a finite-state stationary Markov chain.}
\end{assumption}
This assumption means that for any slot $k$, the current helper-CPU state $C_{k}$ only depends on the previous random state $C_{k-1}$ and is independent of the past states $\{C_1,\cdots, C_{k-2}\}$. Let $P_{00}$ and $P_{11}$ denote the busy-to-busy and idle-to-idle transition probabilities, respectively. Then the busy-to-idle and idle-to-busy transition probabilities, denoted by $P_{01}$ and $P_{10}$, are: $P_{01}=1-P_{00}$ and $P_{10}=1-P_{11}$, respectively. Last, the helper is assumed to reserve a $Q$-bit buffer for storing the offloaded data before processing in the CPU.
\begin{figure}[t!]
  \centering
  \subfigure[Helper-CPU model]{
  \begin{minipage}{6cm}
  \centering
  \includegraphics[width=6cm]{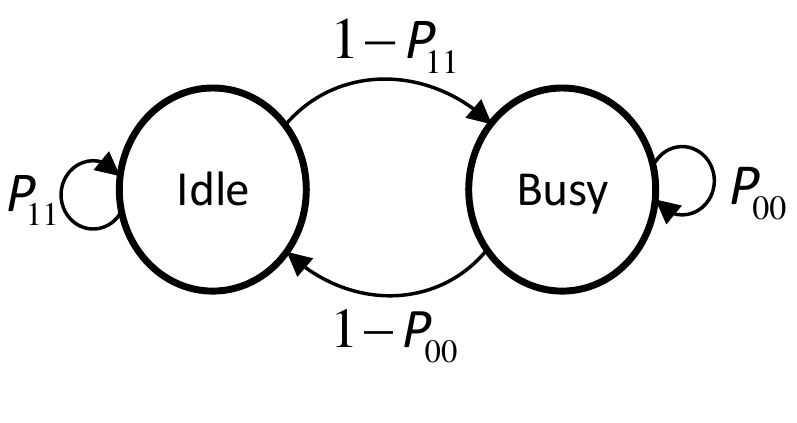}\label{Fig:helper-CPU}
  \end{minipage}
  }
  \subfigure[Channel model]{
  \begin{minipage}{6cm}
  \centering
  \includegraphics[width=6cm]{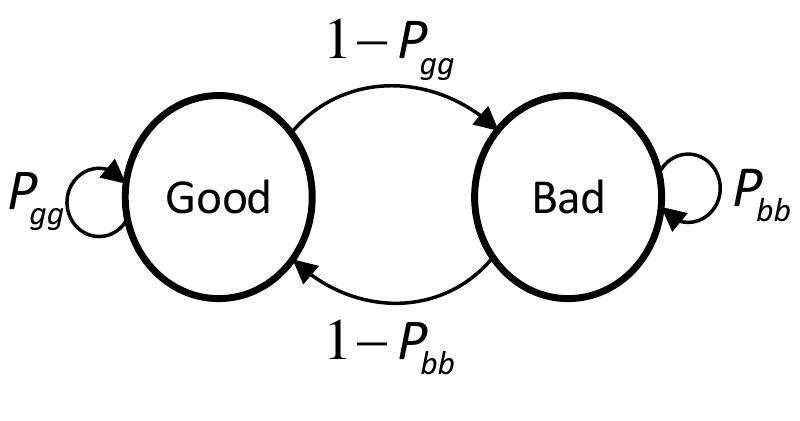}\label{Fig:channel}
  \end{minipage}
  }
  \caption{Markov-chain models of helper-CPU and channel.}\label{Fig:Helper_channel_prf}
\end{figure}
%As shown in Fig.~\ref{Fig:Helper_channel_prf}, let $P_{bb}$ and
\subsection{Models of Local Computing and Computation Offloading}\label{Sec:ModLocOff}
Let $L_k$ denote the remaining data size for processing at the beginning of slot $k$. In this slot, the user offloads $u^{(\rm{of})}_k$-bit data to the helper and computes $u^{(\rm{lo})}_k$-bit data using its local CPU.

First, consider local computing. Let  $f_k$ denote the adjustable CPU-cycle frequency of the user during slot $k$, and $w$ the number of CPU cycles required for computing $1$-bit input-data at the user. Given a fixed data size for computing in a slot, operating at a constant CPU-cycle frequency within the slot is most energy-efficient for local computing \cite{prabhakar2001energy}. Thus the adjustable CPU-cycle frequency of the user is chosen as $f_k=w u_k^{(\rm{lo})}/t_0$.  Following the practical model in \cite{chandrakasan1992low}, the energy consumption of each CPU cycle can be modeled by $E^{(\rm{cyc})}=\gamma f_k^2$, where $\gamma$ is a constant depending on the circuit. The energy consumption for local computing during slot $k$, denoted by $E_k^{(\rm{lo})}$, follows
\begin{equation}
\text{(Local-computing energy consumption)}\quad E_k^{(\rm{lo})}=w \gamma f_k^2 u_k^{(\rm{lo})}= \alpha (u_k^{(\rm{lo})})^3,
\end{equation} where $\alpha=\gamma w^3/t_0^2$.

Next, consider computation offloading. Let $\{H_1, \cdots, H_K\}$ denote the random process of the channel (power) gain between the user and helper. The channel state is modeled by the widely-used Gilbert-Elliott model (see e.g., \cite{zafer2009minimum}), where the channel with the power gain above or below a given threshold is labeled as ``good" or ``bad", with the average channel power gains denoted by $g$ and $b$, respectively. Mathematically, $h_k\in\mathcal{H}\overset{\triangle}{=}\{g,b\}$.
\begin{assumption}[Channel Dynamics]\label{As:Channel}\emph{As shown in Fig.~\ref{Fig:channel}, the process of the random channel states, $\{H_k\}$ for $k=1,\cdots, K$, is a finite-state stationary Markov chain.}
\end{assumption}
\begin{figure}[t!]
  \centering
  \subfigure{
  \begin{minipage}{5.5cm}
  \centering
  \includegraphics[width=5.5cm]{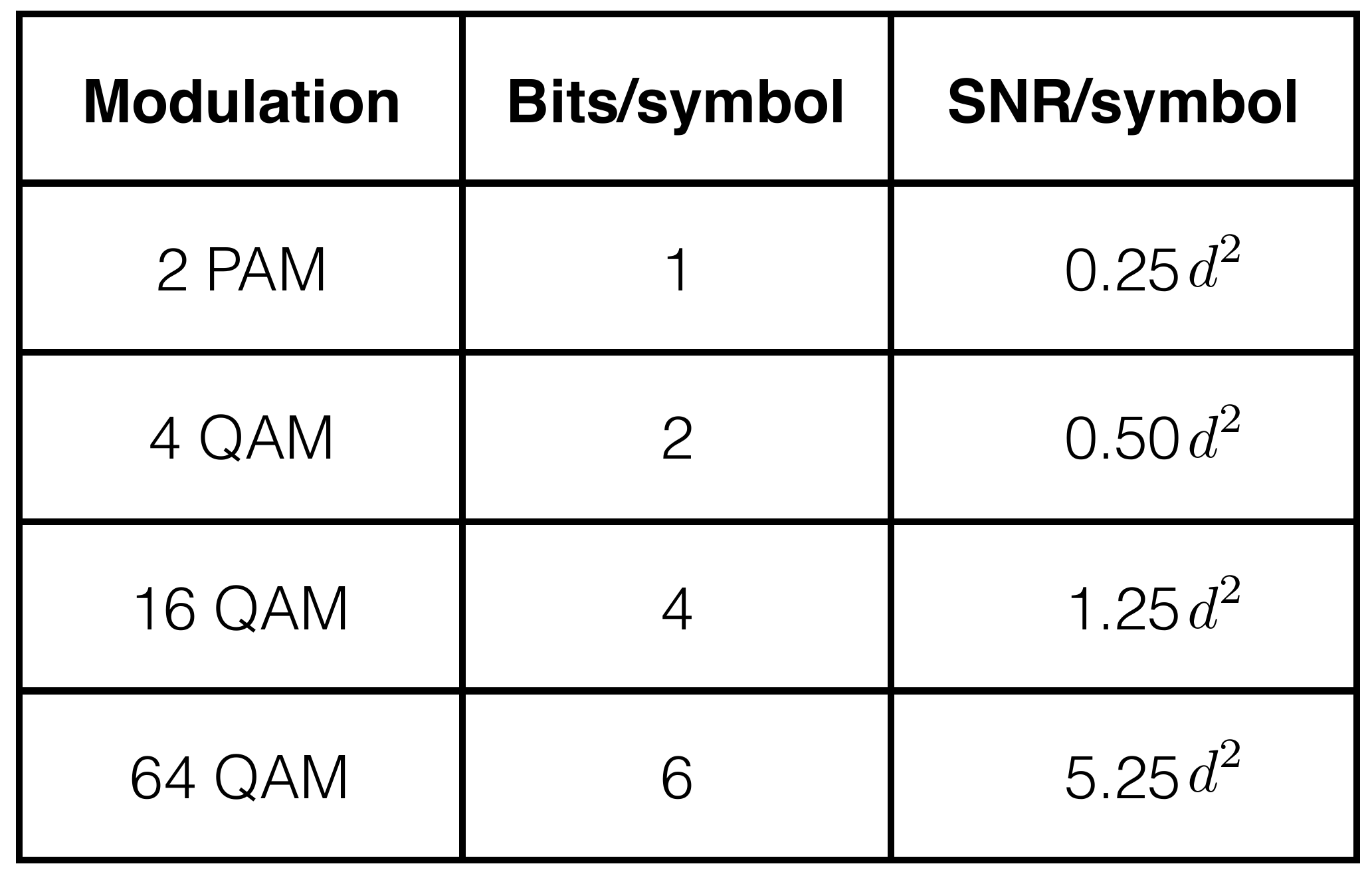}
  \end{minipage}
  }
  \subfigure{
  \begin{minipage}{5.5cm}
  \centering
  \includegraphics[width=5.5cm]{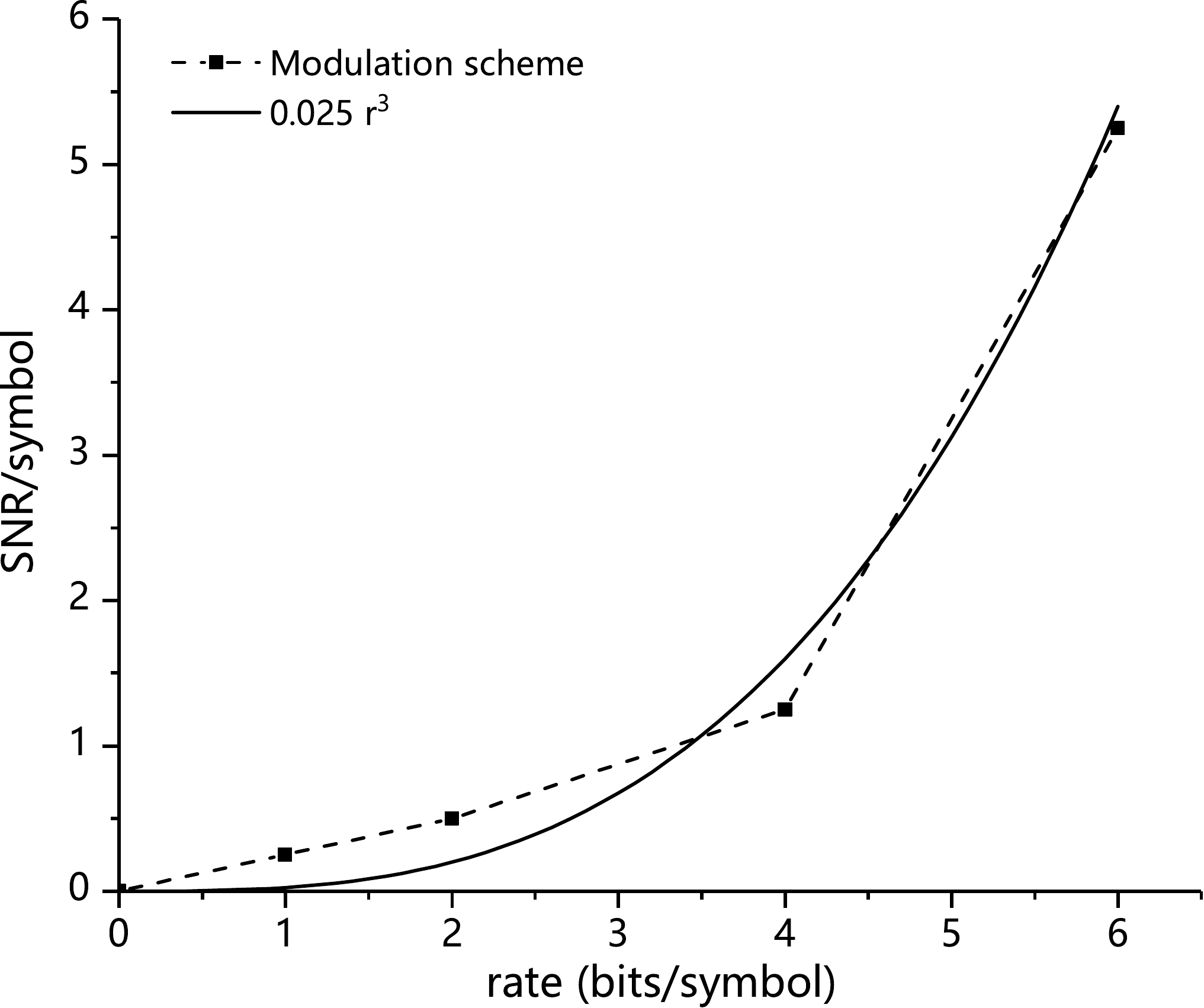}
  \end{minipage}
  }
  \caption{Modulation scheme given in the table is considered in \cite{neely2005dynamic}, where $d$ represents the minimum distance between signal points and SNR is short for the signal-to-noise ratio. The corresponding plot shows $0.025r^3$ to the scaled piecewise linear power-rate curve.}\label{Fig:Monomial_test}
\end{figure}
Let $P_{gg}$ and $P_{bb}$ denote the transition probability from good-to-good and bad-to-bad respectively, then the transition probabilities from good-to-bad and from bad-to-good, denoted by $P_{gb}$ and $P_{bg}$, are given as: $P_{gb}=1-P_{gg}$ and $P_{bg}=1-P_{bb}$, respectively. Following the empirical model in the literature  \cite{zafer2009minimum,zhang2013energy,ko2017live}, the transmission power in slot $k$, denoted by $p_k$, can be modeled by a \emph{monomial} function with respect to the achievable transmission rate $r_k$ (in bits/s):
\vspace{-5pt}
\begin{equation}\label{Eq:EgyOff}
\text{(Monomial offloading power)}\quad p_k=\lambda_0 \frac{(r_k)^m}{h_k},
\end{equation}
where $\lambda_0$ denotes the energy coefficient incorporating the effects of bandwidth and noise power, and $m$ is the monomial order determined by the modulation-and-coding scheme. This monomial order is a positive integer, taking on values in the typical range of $2\le m\le 5$. Specifically, considering the coding scheme for the targeted error probability less than $10^{-6}$  \cite{neely2005dynamic}, the monomial order of $(m=3)$ can approximate transmission energy consumption as shown in Fig.~\ref{Fig:Monomial_test}. Then for tractability, the offloading energy consumption in slot $k$ can be modeled by the following monomial function of the offloaded data size $u_k^{(\rm{of})}$:
\vspace{-5pt}
\begin{equation}\label{Eq:EgyOff}
\text{(Monomial offloading energy consumption)}~ E_k^{(\rm{of})}=p_k t_0=
% \lambda_0 \frac{\l(\ell_k^{(\rm{of})}\r)^3}{t_0^ 3 G_k} t_0=
 \lambda \frac{(u_k^{(\rm{of})})^3}{h_k},
\end{equation}
where $\lambda=\lambda_0/t_0^{2}$.

\subsection{Model of Opportunistic Computation Offloading}
%At the beginning of the first slot, a
Assume that the random helper CPU and channel states are \emph{independent} and the helper has the knowledge of channel-and-CPU distributions as given in Assumption~\ref{As:CPU} and \ref{As:Channel}. Using the user's information by feedback including the input-data size $D$, computation deadline $T$ and local-computing parameters $\{w, \gamma\}$, the helper computes the control policy and then controls the offloading process so as to minimize user-energy consumption.
\begin{assumption}[Causal State Information]\emph{The helper has \emph{causal} information of the helper CPU and \emph{channel state information} (CSI), where CSI is obtained by feedback.}
\end{assumption}

The opportunistic computation offloading is elaborated as follows. Consider slot $k$. At the beginning of this slot, the user observes the helper CPU and channel states, and makes decisions on the locally-computed and offloaded data sizes in this slot. The remaining data size in the next slot, denoted by $L_{k+1}\in \mathcal{L}\overset{\triangle}{=}\{1,\cdots, D\}$, evolves as:
\begin{equation}\label{Eq:RemainData}
\text{(Remaining data size)}\quad L_{k+1}=L_k-u_k^{(\rm{of})}-u_k^{(\rm{lo})}, \qquad \forall k=1, \cdots, K-1,
\end{equation}
with $L_1=D$. For the offloaded data, it is assumed to be first stored in the helper buffer and then fetched to the helper CPU for computing. Let $Q_k\in \mathcal{Q}\overset{\triangle}{=}\{1,\cdots, Q\}$ denote the buffered data size at the beginning of slot $k$. Assume that the helper CPU has \emph{unlimited}   computation capacity, such that it can compute all the buffered data if it is idle (not occupied by a primary task). Mathematically, if $c_k=1$,  it has $Q_{k+1}=0$. Otherwise, if the helper CPU is busy (i.e., $c_k=0$), all the newly-offloaded data is stored in the buffer in the next slot. Then the buffered data size evolves as follows:
\vspace{-5pt}
\begin{equation}\label{Eq:Buffer}
\text{(Buffered data size)}\quad Q_{k+1}=
\begin{cases}
  Q_{k}+u_k^{(\rm{of})}, &c_k=0,\\
0, &c_k=1,
\end{cases}\qquad \forall k=1,  \cdots, K-1,
\end{equation}
with $Q_1=0$,  which can be rewritten as $Q_{k+1}\!=\!(Q_k+u_k^{(\rm{of})})(1-c_k)$.   To guarantee that all the offloaded data is computed by the deadline, we propose a \emph{demand-computing} scheme as follows.
\begin{definition}[Demand-Computing]\emph{In the last slot $K$, if the helper CPU is busy, demand-computing refers to the scheme at the user that receives feedback from the helper on the un-computed data as offloaded by the user and then computes it completely using the local CPU by adjusting its frequency.}
\end{definition}

\vspace{-10pt}
\section{Problem Formulation}\label{Sec:Problem Formulation}
In this section, the energy-efficient stochastic control of opportunistic computation offloading is formulated as a finite-horizon MDP problem for minimizing the expected user-energy consumption. Solving the problem yields policies for the helper to control the offloading process.

 The formulated problem consists of the following components.
\subsubsection{State space}\label{Def:State_spc}
Let $x_k=(x_k^{(C)},x_k^{(H)},x_k^{(L)},x_k^{(Q)})$ denote the system state in each slot $k$, where $x_k^{(C)}$, $x_k^{(H)}$, $x_k^{(L)}$ and $x_k^{(Q)}$ redenote the helper-CPU state $c_k$, channel state $h_k$, remaining data size $L_k$, and buffered data size $Q_k$, respectively. Then system state space is the product space $\mathcal{X}=\mathcal{C}\times\mathcal{H}\times\mathcal{L}\times\mathcal{Q}$.

\subsubsection{Action space}\label{Def:ContSpace}
Let $u_k=(u_k^{(\rm{of})},u_k^{(\rm{lo})})$ denote the action in slot $k$ with  $u_k^{(\rm{of})}\ge 0$ and $u_k^{(\rm{lo})}\ge 0$. The action space is $\mathcal{L}\times\mathcal{L}$.  Given a specific system state $x_k$, the \emph{feasible} action $u_k$ depends on the state, that is, $u_k\in \mathcal{U}(x_k)$ for all $x_k\in \mathcal{X}$ where $\mathcal{U}(x_k)$ is the space of all feasible actions in slot $k$ satisfying the following constraints.
First,  the data constraints require:
\vspace{-5pt}
\begin{equation}\label{Eq:DataCons}
\text{(Data constraints)}~~u_k^{(\rm{of})}+u_k^{(\rm{lo})}\le x_k^{(L)}, \quad\forall k=1,\cdots,K-1.
\end{equation}
Next, in the last slot $K$, the demand-computing enforces the following constraint:
\vspace{-5pt}
\begin{equation}\label{Eq:DemandCon}
\text{(Demand-computing constraint)}~~\begin{cases}
   u_K^{(\rm{of})}=0, ~u_K^{(\rm{lo})}=x_K^{(L)}+x_K^{(Q)},  & ~~~x_K^{(C)}=0, \\
   u_K^{(\rm{of})}+u_K^{(\rm{lo})}= x_K^{(L)},  & ~~~x_K^{(C)}=1.
  \end{cases}
\end{equation}
Last, the finite helper buffer leads to:
\vspace{-5pt}
\begin{equation}\label{Eq:BufferCons}
\text{(Buffer constraints)}~~\begin{cases}
   u_k^{(\rm{of})}\le Q-x_k^{(Q)},  & ~~~x_k^{(C)}=0, \\
   u_k^{(\rm{of})}\le x_k^{(L)},  & ~~~x_k^{(C)}=1,
  \end{cases}
\quad  \forall k.
\end{equation}
Based on the above discussion, the feasible action space can be summarized as
\vspace{-5pt}
\begin{equation}\label{Eq:ContSpace}
\mathcal{U}(x_k)=\begin{cases}
\l\{u_k~\big|~\eqref{Eq:DataCons},\eqref{Eq:BufferCons}, u_k^{(\rm{of})}\ge 0, u_k^{(\rm{lo})}\ge 0\r\}, &~~k=1,\cdots,K-1;\\
\l\{u_k~\big|~\eqref{Eq:DemandCon},\eqref{Eq:BufferCons},u_k^{(\rm{of})}\ge 0, u_k^{(\rm{lo})}\ge 0\r\}, &~~k=K.
\end{cases}
\end{equation}
\subsubsection{State transition probability}
 The state transition probability, denoted by $\Pr(x_{k+1}| x_k, u_k)$, is the probability that the system will be in the state $x_k$ in slot $(k+1)$, given the current state $x_k$ and the action $u_k$. Since $x_k^{(C)},x_k^{(H)},x_k^{(L)}$ and  $x_k^{(Q)}$ are independent, the \emph{controllable} state transition probability can be written as:
\begin{align}\label{Eq:Transition}
&\Pr(x_{k+1}\big| x_k, u_k)=\Pr(x_{k+1}^{(C)}\big|x_k^{(C)})\times \Pr(x_{k+1}^{(H)}\big|x_k^{(H)}) \nn \\
&\quad\times \mathbb{I}\l[x_{k+1}^{(L)}=x_k^{(L)}-u_k^{(\rm{of})}-u_k^{(\rm{lo})}\r]\times \mathbb{I}\l[x_{k+1}^{(Q)}=(x_k^{(Q)}+u_k^{(\rm{of})})(1-x_k^{(C)})\r], ~\forall~ x_{k+1}\in \mathcal{X},
\end{align}
where $\mathbb{I}[\cdot]$ denotes the indicator function. According to \eqref{Eq:RemainData} and \eqref{Eq:Buffer}, the state transition probability is positive only when $x_{k+1}^{(L)}=x_k^{(L)}-u_k^{(\rm{of})}-u_k^{(\rm{lo})}$ and $x_{k+1}^{(Q)}=(x_k^{(Q)}+u_k^{(\rm{of})})(1-x_k^{(C)})$, simultaneously, for which $\Pr(x_{k+1}\big| x_k, u_k)=\Pr(x_{k+1}^{(C)}\big|x_k^{(C)})\times \Pr(x_{k+1}^{(H)}\big|x_k^{(H)})$.
\subsubsection{Cost function, control rule and policy}
The cost function is specified by the function $R(x_k, u_k)$, giving the total energy consumption in each slot $k$:
\begin{equation}\label{Eq:PerCost}
(\text{Cost function})\quad R(x_k, u_k)=\alpha (u_k^{(\rm{lo})})^3 +\lambda \frac{(u_k^{(\rm{of})})^3}{x_k^{(H)}}.
\end{equation}
Let the function $\pi_k$ denote the control rule that maps state $x_k$ to the action $u_k=\pi_k(x_k)$ such that $\pi_k(x_k)\in \mathcal{U}(x_k)$. An \emph{admissible} policy is a sequence of control rules in each slot, $\boldsymbol{\pi}=\{\pi_1, \cdots, \pi_K\}$. The set of all admissible policies is denoted by $\boldsymbol{\Pi}$.
%\begin{assumption}
%\emph{The data bit is arbitrary splittable.}
%\end{assumption}
%This assumption is commonly used in literature \cite{zafer2009minimum}. Thus, we relax the discrete states of remaining data size and buffer data size as the \emph{continuous} states, and also the actions $\{u_{k}\}$ as continuous actions.

Then given any initial state $x_1$, the optimization problem for the minimum expected sum user-energy consumption can be formulated as:
\begin{equation*}
\textbf{(P1)}\qquad
\begin{aligned}
\min_{\boldsymbol{\pi}\in\boldsymbol{\Pi}} ~ &\E \l[ \sum_{k=1}^K  R(X_k, \pi_k(X_k)) \bigg| x_1\r],
\end{aligned}
\end{equation*}
where the expectation is taken with respect to the system state vector $\boldsymbol{X}_k=\{X_1, \cdots, X_K\}$.  Denote  the minimum expected sum user-energy consumption as $J^{*}(x_1)$. To solve Problem P1, note that the optimization for $\pi_k(x_k)$ in different slots cannot be performed independently, since the action in slot $k$ affects the system transition probability in subsequent slots [see \eqref{Eq:Transition}] and thus the future energy cost. This difficulty is overcome in the following lemma according to the \emph{Principle of Optimality} \cite{bertsekas1995dynamic}.
\begin{lemma}[DP Approach]\emph{For any initial state $x_1=(x_1^{(C)},x_1^{(H)}, D, 0)$ where $x_1^{(C)}\in\mathcal{C}$ and $x_1^{(H)}\in\mathcal{H}$, the minimum expected sum user-energy consumption  $J^{*}(x_1)$ is equal to $J_1(x_1)$, which can be recursively solved based on the Bellman Equation, starting from $J_K(x_K)$, $J_{K-1}(x_{K-1})$ to $J_1(x_1)$ by evaluation all possible states in each slot $k=1, \cdots, K$. Specifically, $J_k(x_k)$ denotes the \emph{cost-to-go} function that gives the minimum expected user-energy cost from slot $k$ to $K$, which is defined as:
\begin{equation}\label{eq:dp_no_contr}
%\textbf{(P2)}\quad
   J_k(x_k) = \begin{cases}
   \min \limits_{u_k \in \mathcal{U}(x_k)} \l\{ R(x_k, u_k) +  \E\l[ J_{k+1}(X_{k+1})|x_k,u_k \r]\r\}, & k < K, \\
   \min \limits_{u_K \in \mathcal{U}(x_K)} R(x_K, u_K),  & k=K,
  \end{cases}
\end{equation}
where $\mathcal{U}(x_k)$ is defined in \eqref{Eq:ContSpace}.
Furthermore, if $u_k^*=\pi_k^*(x_k)$ is the optimal control rule for solving  \eqref{eq:dp_no_contr} for each $x_k$ and $k$, the optimal policy for solving Problem P1 is given by $\boldsymbol{\pi}^*=\{\pi_1^*, \cdots, \pi_K^*\}$.
}
\end{lemma}

Denote the problem for deriving $J_1(x_1)$ [or $J^*(x_1)$] as Problem P2. The resultant optimal computation-offloading policy for solving Problem P2 can be computed numerically, but the result yields little insight into the policy structure, which is important for practical policy computation and design. Thus, we quantify the policy structure and develop low-complexity schemes for offloading control in the following sections. The analysis involves solving convex problems using the Lagrange method. For simplicity, we assume in the analysis  that data is \emph{continuously divisible}, following a common approach in the literature (see e.g., \cite{zafer2009minimum}). In practice, the derived data sizes are  rounded to bits.

\section{Computation Offloading Control: No Helper Buffer}\label{Sec:ContWoBuf}
In this section, we derive the optimal offloading-control policy for the case where the helper has no buffer for storing the offloaded data. To this end, we first derive the optimal computation offloading policy for the last slot and then build on the result to obtain the optimal policy for multi-slots.
%for the multi-slots.

In this case, due to the zero buffer, the state of buffered data size is $x_k^{(Q)}=0, \forall k$. Consider the offloading control in the last slot $K$ given the system state $x_K$.  The user can offload partial data if the helper CPU is idle, but can only perform local computing if the helper CPU is busy. Thus the problem for $J_K(x_K)$ can be simplified as follows.
\begin{equation}
\textbf{(P3)} \quad
   J_K(x_K) = \begin{cases}
   \min \limits_{u_K \in \mathcal{U}(x_K)} \l\{\alpha (u_K^{(\rm{lo})})^3+\lambda \dfrac{(u_K^{(\rm{of})})^3}{x_K^{(H)}} \r\}, & x_K^{(C)}=1, \\
    \alpha (x_K^{(L)})^3 ,  & x_K^{(C)}=0,
  \end{cases} \notag
\end{equation}
where $\mathcal{U}(x_K)$ is defined in \eqref{Eq:ContSpace}, which can be explicitly written as
\begin{equation}\label{Eq:UK}
\mathcal{U}(x_K)=\{u_K~|~u_K^{(\rm{lo})}\ge0, u_K^{(\rm{of})}\ge0, u_K^{(\rm{lo})}+u_K^{(\rm{of})}=x_K^{(L)}\}.
\end{equation}
%Note that policy $u_K^{(\rm{of})}$ and $u_K^{(\rm{lo})}$ are integers. To derive closed-form solutions, we first relax them as continuous variables to solve, and then round the solutions to the nearest integers. Since the rounding procedure does not affect the structure of the policy, we omit it for brevity.
It is easy to prove that Problem P3 is a convex optimization problem. Applying the Lagrange method yields the optimal solution given the following lemma.
\begin{lemma}[Optimal Control for Last Slot]\label{Lem:ConK}\emph{For the last slot, the optimal locally-computed and  offloaded data sizes $u_K^{(\rm{lo})*}$ and $u_K^{(\rm{of})*}$ for solving Problem P3 are given by
\begin{equation}
u_K^{(\rm{lo})*}=\left( 1+\sqrt{\frac{\alpha}{\lambda}}x_K^{(C)}\sqrt{x_K^{(H)}}\right )^{-1}\!x_K^{(L)}\quad\text{and}\quad
u_K^{(\rm{of})*}=\l(1+\sqrt{\frac{\lambda}{\alpha x_K^{(H)}}}\r)^{-1}\!x_K^{(L)}x_K^{(C)}.
\end{equation}
The minimum user-energy consumption $J_K(x_K)$ is
\begin{equation}\label{Eq:ConKEgy}
J_K(x_K)=\alpha (x_K^{(L)})^3 \left( 1+\sqrt{\frac{\alpha}{\lambda}}x_K^{(C)}\sqrt{x_K^{(H)}} \right )^{-2}.
\end{equation}}
\end{lemma}
\begin{proof}
See Appendix~\ref{app:conK}.
\end{proof}
\begin{remark}[Channel-Aware Proportional Offloading]\emph{
Lemma~\ref{Lem:ConK} shows that when the helper CPU is idle, the offloaded data size is \emph{proportional} to the remaining data size with the scaling factor $\l(1+\sqrt{\lambda/(\alpha x_K^{(H)})}\r)^{-1}$, which is determined by the channel state $x_K^{(H)}$. In particular, the user offloads more data when the channel is in a good state. Moreover, one can observe from \eqref{Eq:ConKEgy} that, compared with no offloading (i.e., only local computing), the current opportunistic offloading can help reduce user's energy consumption by the proportional factor $\left( 1+\sqrt{\alpha x_K^{(H)}/\lambda} \right )^{-2}$ when the helper CPU is idle.}
\end{remark}

Lemma~\ref{Lem:ConK} gives the minimum energy consumption in the last slot, which is expressed in a compact form accounting for  different system states. This key observation facilitates the stochastic computation offloading policy design. Specially, using Lemma~\ref{Lem:ConK} and applying backward induction for solving Problem P2, the optimal offloading policy can be derived as shown below.

\begin{theorem}[Optimal Policy for Multi-Slots]\label{Theo:OptW/oBuf}
\emph{Consider that the helper has no buffer. For $k=1, \cdots K$, the optimal computation offloading policy allocates $u_k^{(\rm{of})*}$ and $u_k^{(\rm{lo})*}$ bits for offloading and local computing, respectively, which satisfy the following relation:
\begin{equation}\label{eq:no_opt_relation}
  \frac{u_k^{(\rm{of})*}}{u_k^{(\rm{lo})*}}=\sqrt{\frac{\alpha}{\lambda}}x_k^{(C)}\sqrt{x_k^{(H)}}.
\end{equation}
Specifically, $u_k^{(\rm{lo})*}$ and $u_k^{(\rm{of})*}$ are given by:
\begin{align}\label{eq:opt_plcy}
\!\!\!\!u_k^{(\rm{lo})*}\!=\!x_k^{(L)}\!\!\l(1\!+\!\frac{1}{\sqrt{S_k(x_k^{(C)}, x_k^{(H)})}}\!+\!\sqrt{\frac{\alpha}{\lambda}}x_k^{(C)}\sqrt{x_k^{(H)}}\!\r)^{-1} \text{and}~
u_k^{(\rm{of})*}\!=\!\sqrt{\frac{\alpha}{\lambda}}x_k^{(C)}\sqrt{x_k^{(H)}}u_k^{(\rm{lo})*},
\end{align}
where $S_k(x_k^{(C)}, x_k^{(H)})$ is defined as
\begin{equation}\label{eq:xi}
 S_k(x_k^{(C)}, x_k^{(H)})\!=\!\begin{cases}
 \begin{split}&\sum\limits_{x_{k+1}^{(C)}\in\{0,1\}}\sum\limits_{x_{k+1}^{(H)}\in\{g,b\}} \Pr(x_{k+1}^{(C)}\big|x_k^{(C)})\Pr(x_{k+1}^{(H)}\big|x_k^{(H)})\\
  &\qquad\quad\l(1+\frac{1}{\sqrt{S_{k+1}(x_{k+1}^{(C)}, x_{k+1}^{(H)})}}+\sqrt{\frac{\alpha}{\lambda}}x_{k+1}^{(C)}\sqrt{x_{k+1}^{(H)}}\r)^{-2}, ~k < K, \end{split} \\
 \infty,  ~\qquad\qquad\qquad\qquad\qquad\qquad\qquad\qquad\qquad\qquad\qquad~\ k=K.
  \end{cases}
\end{equation}
The corresponding minimum expected user-energy consumption is
\begin{equation}\label{eq:J_1}
  J_1(x_1)=\alpha D^3 \l(1+\frac{1}{\sqrt{S_1(x_1^{(C)}, x_1^{(H)})}}+\sqrt{\frac{\alpha}{\lambda}}x_1^{(C)}\sqrt{x_1^{(H)}}\r)^{-2}.
\end{equation}
}
%\begin{proof}
%\emph{See Appendix~\ref{app:opt_data_scheduling}.}
%\end{proof}
\end{theorem}
\begin{proof}
See Appendix~\ref{app:opt_data_scheduling}.
\end{proof}

Theorem~\ref{Theo:OptW/oBuf} shows that in each slot, the optimal computation offloading control is not only determined by the current state $x_k$, but also the future energy cost from slot $(k+1)$ to $K$ by the term of $S_k(x_k^{(C)}, x_k^{(H)})$. Specifically, a smaller $S_k(x_k^{(C)}, x_k^{(H)})$ indicates less expected energy consumption for computing per bit data in future slots. In this case, the user should reduce both the locally-computed and offloaded data sizes. Moreover, if the helper-CPU state is idle in slot $k$, the fraction between the optimal offloaded and locally-computed data sizes is determined by the channel state. Specifically,  ${u_k^{(\rm{of})*}}/{u_k^{(\rm{lo})*}}=\sqrt{\alpha x_k^{(H)}/\lambda}$. This result implies that given the sum processed data in one slot, more data should be offloaded if the local-computing complexity is higher and channel is better, which is aligned with intuition.
%Furthermore, one can observe from \eqref{eq:J_1} that as the input-data size increases, the minimum expected energy consumption, $J_1(x_1)$, monotonically increases and scales with a factor of $D^3$.
%The initial states with the helper-CPU being idle and good channel can reduce the energy consumption when the deadline requirement is stringent (i.e., small $K$), but the effect attenuates for a long deadline (i.e., large $K$).

%function $A(x_k)$ which characterizes the effects of future system states .
\begin{remark}[Low-Complexity Algorithm]
\emph{The computational complexity for the direct DP approach depends on the dimensions of state space, action space and finite horizon. As the dimensions of the state space and action space are $\mathcal{O}(2\times 2\times D)$ and $\mathcal{O}(D^2)$, respectively, the total computational complexity has the order of $\mathcal{O}(16D^2\times D^2 \times K)$, which is impractical for large $K$ and $D$. This curse of dimensionality \cite{bertsekas1995dynamic} is overcome by deriving the closed-form expression for the optimal solution as in Theorem~\ref{Theo:OptW/oBuf}, avoiding the exhausted search for the optimal control in the action space in each slot.  Specifically, the optimal control is directly determined by the observed system state $x_k$ and $S_k(x_k^{(C)}, x_k^{(H)})$, which essentially depends on the dimensions of helper CPU and channel state spaces and the number of slots. Thus the computational complexity for the proposed algorithm is the complexity for calculating $S_1(x_1^{(C)}, x_1^{(H)})$, which is $\mathcal{O}\l(4 K\r)$ and much lower than that of the brute-force DP approach.
 }
\end{remark}

\section{Computation Offloading Control: Large Helper Buffer}\label{Sec:LargeBuffer}
In this section, we consider the design of computation offloading policy for the case where the helper has a large buffer with a size $Q$ sufficient for storing  the maximum offloaded  data $D$: $Q\geq D$. The direct application of iterative DP results in complexity linearly scaling with $Q$. To tackle this challenge, a low-complexity policy with close-to-optimality is designed based on approximating the minimum expected energy cost function and found to have a similar structure as the zero-buffer counterpart studied in the preceding section.

\subsection{Optimal Policy Computation}\label{Sec:LargeOpt}
In this subsection, we derive the procedure for computing the optimal offloading policy using DP, which facilitates the low-complexity policy design in the sequel. To begin with, consider computation offloading in the last slot $K$. Note that compared with the case of no buffer (see Section~\ref{Sec:ContWoBuf}), the current case requires  an additional constraint that, if the helper CPU is busy and $x_K^{(Q)}$-bit uncomputed data is buffered, the demand-computing requires local computing to finish all the remaining computation, yielding the energy consumption of $\alpha (x_K^{(L)}+x_K^{(Q)})^3$. Thus, Problem P3 can be rewritten as
\begin{equation}
 (\textbf{P4}) \quad J_K(x_K)=\begin{cases} \min \limits_{u_K\in \mathcal{U}(x_K)} \l\{\alpha(u_K^{(\rm{lo})})^3+\lambda  \dfrac{(u_K^{(\rm{of})})^3}{x_K^{(H)}} \r\}, & x_K^{(C)}=1, \\
    \alpha (x_K^{(L)}+x_K^{(Q)})^3 ,  & x_K^{(C)}=0, \notag
   \end{cases}
\end{equation}
where $\mathcal{U}(x_K)$ is given in \eqref{Eq:UK}. Following the same procedure as for solving Problem P3 leads to the optimal solution to Problem P4 as follows.

\begin{lemma}[Optimal Policy for Last Slot]\label{Lem:LargK}
\emph{For the $K$-th slot, the optimal locally-computed and offloaded  data sizes for solving Problem P4 are given by:
\begin{align}\label{eq:inf_JK_plcy}
u_K^{(\rm{lo})*}&=\left(1+ \sqrt{\dfrac{\alpha}{\lambda}}x_K^{(C)}\sqrt{x_K^{(H)}} \right )^{-1} \! \l[x_K^{(L)}+(1-x_K^{(C)})x_K^{(Q)}\r],\\
u_K^{(\rm{of})*}&=\l(1+\sqrt{\frac{\lambda}{\alpha x_K^{(H)}}}\r)^{-1}\!x_K^{(C)}x_K^{(L)}.
\end{align}
The corresponding minimum user-energy consumption  is
\begin{equation}\label{eq:inf_JK}
  J_K(x_K)=\alpha \l[x_K^{(L)}+(1-x_K^{(C)})x_K^{(Q)}\r]^3 \left(1+ \sqrt{\frac{\alpha}{\lambda}}x_K^{(C)}\sqrt{x_K^{(H)}} \right )^{-2}.
\end{equation}
}
\end{lemma}
Comparing Lemma~\ref{Lem:LargK} and Lemma~\ref{Lem:ConK}, we can observe that offloading policies in the last slot for the two scenarios of large and zero buffers are the same, while the local computing policy for the case of large buffer should compute more data than the zero-buffer counterpart when the helper CPU is busy, due to the demand-computing constraint.

Though the last-slot policy has a closed form, the optimal multi-slot policy is intractable due to the helper buffer. Based on DP, the backward iterative computation of the minimum expected user-energy cost in slot $(K-1)$ is given as
\begin{equation}
 J_{K-1}(x_{K-1}) = \min \limits_{u_{K-1} \in \mathcal{U}(x_{K-1})} \l\{R(x_{K-1}, u_{K-1}) +  \E\l[ J_{K}(X_{K})|x_{K-1},u_{K-1} \r]\r\},
\end{equation} where the action space $\mathcal{U}(x_{K-1})$ is defined in \eqref{Eq:ContSpace}. In particular, using \eqref{Eq:Transition} and Lemma~\ref{Lem:LargK}, the expected conditional energy cost in slot $(K-1)$ is obtained as
\begin{align}\label{Eq:LargEgyK}
&\E[J_K(X_K)|x_{K-1},u_{K-1}] \nn \\
&=\sum\limits_{x_{K}^{(C)}\in\{0,1\}}\sum\limits_{x_{K}^{(H)}\in\{g,b\}} \Pr(x_{K}^{(C)}\big|x_{K-1}^{(C)})\Pr(x_{K}^{(H)}\big|x_{K-1}^{(H)})  \frac{\alpha F(x_{K-1}, u_{K-1}, x_{K}^{(C)})}{\left( 1+\sqrt{\dfrac{\alpha}{\lambda}}x_K^{(C)}\sqrt{x_K^{(H)}} \right)^2},
\end{align}
where the function $F(\cdot)$ is defined as
\begin{equation}\label{EqF}
\!F(x_{K-1}, u_{K-1}, x_{K}^{(C)})\!=\!\l[(x_{K-1}^{(L)}-u_{K-1}^{(\rm{lo})}-u_{K-1}^{(\rm{of})})+(1-x_K^{(C)}) (1-x_{K-1}^{(C)}) (x_{K-1}^{(Q)}+u_{K-1}^{(\rm{of})})\r]^3.\!\!\!\!
\end{equation}
One can observe from \eqref{Eq:LargEgyK} and \eqref{EqF}  that the total computed data in the last slot  depends on the random helper-CPU state $X_K^{(C)}$, which, however, is equal to $(x_{K-1}^{(L)}-u_{K-1}^{(\rm{lo})}-u_{K-1}^{(\rm{of})})$ for the counterpart of zero buffer. This is due to that for the large buffer, the user can offload data in slot $(K-1)$ even when the helper CPU is busy (i.e, $x_{K-1}^{(\rm{C})}=0$) and the buffered data should be computed locally in the last slot if the helper CPU is busy (i.e, $x_{K}^{(\rm{C})}=0$). It is difficult to derive  a closed-form expression for the optimal offloading-control policy for the $(K-1)$-th slot. This is also true for all of the $1$ to $(K-1)$-th slots. Alternatively, computing the optimal policy has to rely on the numerical method of DP, based on iterations similar to that in \eqref{eq:dp_no_contr}. The complexity, however, is high as discussed in the following remark.

\begin{remark}[Curse of Dimensionality]\emph{The computational complexity for the numerical of DP is similar to the case of zero buffer. Due the large buffer (i.e., $Q>D$), the state space of the buffered data size can be bounded as $\mathcal{D}$. Thus, the dimensions of the state space and action space are $\mathcal{O}(2\times 2\times D^2)$ and $\mathcal{O}(D^2)$, respectively. Then the total computational complexity for the optimal policy can be derived as $\mathcal{O}(16D^4\times D^2 \times K)$, which incurs prohibitive computation complexity if $K$ and $D$ are large.}
\end{remark}

\subsection{Sub-Optimal Policy Design}
To address the above complexity issue in computing the  optimal policy, we design a tractable sub-optimal offloading-control policy by approximation the expected future energy-cost function as follows.

\subsubsection{Approximate  Energy Cost Function}
To approximate the expected conditional energy cost
in \eqref{Eq:LargEgyK}, the function $\E[J_K(X_K)|x_{K-1},u_{K-1}]$ is firstly bounded as follows.
\begin{lemma}
\label{Lem:Lb_cost}
\emph{The expected conditional energy-cost function, $\E[J_K(X_K)|x_{K-1},u_{K-1}]$, in \eqref{Eq:LargEgyK} can be lower-bounded as
\begin{equation}\label{eq:lb_cost_large2}
  \E[J_K(X_K)|x_{K-1},u_{K-1}] \ge \alpha \widetilde{F}(x_{K-1},u_{K-1}) S(x_{K-1}^{(C)},x_{K-1}^{(H)}),
\end{equation}
where the two functions  $\widetilde{F}(\cdot)$ and $S(\cdot)$ are defined as
\begin{align}
\widetilde{F}(x_{K-1},u_{K-1})\!&=\!\l[(x_{K-1}^{(L)}\!-\!u_{K-1}^{(\rm{lo})}\!-\!u_{K-1}^{(\rm{of})})+V(x_{K-1}^{(C)})(x_{K-1}^{(Q)}+u_{K-1}^{(\rm{of})})\r]^3, \label{Eq:TiF}\\
S(x_{K-1}^{(C)},x_{K-1}^{(H)})&=\sum_{x_K^{(C)}\in\mathcal{C}}\sum_{x_K^{(H)}\in\mathcal{H}} \Pr(x_K^{(C)}|x_{K-1}^{(C)})\Pr(x_K^{(H)}|x_{K-1}^{(H)})\left( \sqrt{\frac{\alpha}{\lambda}}x_K^{(C)}\sqrt{x_K^{(H)}}+1 \right )^{-2},\label{eq:Large_xi_K-1}
 \end{align}
and the function $V(\cdot)$ is defined as
\begin{equation}\label{Eq:V}
V(x_{K-1}^{(C)})= (1-x_{K-1}^{(C)})P_{00}.
\end{equation}
The equality in \eqref{eq:lb_cost_large2} holds if the user has non-causal information of channel and helper-CPU states.
}
\end{lemma}
\begin{proof}
See Appendix~\ref{App:Lb_cost}.
\end{proof}

It can be observed that the function of $\widetilde{F}(x_{K-1},u_{K-1})$ in \eqref{Eq:TiF} has a similar form with $F(x_{K-1}, u_{K-1}, x_{K}^{(C)})$ in \eqref{EqF}. The key difference is that, $(1-x_{K-1}^{(C)})(1-x_K^{(C)})$ in \eqref{EqF} is replaced with $V(x_{K-1}^{(C)})$ in \eqref{Eq:TiF}, which helps derive the closed-form expression for the lower-bound of  $\E[J_K(X_K)|x_{K-1},u_{K-1}]$ as shown in \eqref{eq:lb_cost_large2}.  Intuitively, $V(x_{K-1}^{(C)})$ can be interpreted as the conditional probability that the buffered data in slot $(K-1)$ will be locally-computed  in slot $K$. In particular, if $x_{K-1}^{(C)}=1$, all the buffered data is computed by the helper within this slot and thus $V(x_{K-1}^{(C)})=0$. On the other hand, if $x_{K-1}^{(C)}=0$, $V(x_{K-1}^{(C)})=P_{00}$ showing that the probability of demand-computing for $(x_{K-1}^{(Q)}+u_{K-1}^{(\rm{of})})$-bit is the conditional busy probability in slot $K$.

Given the objective of energy minimization, though it is desirable to obtain a tight and tractable upper bound on the expected energy cost, finding such a bound is difficult. For this reason, we instead approximate the function using its lower bound in Lemma 4 as follows and verify the close-to-optimality by simulation.
\begin{equation}\label{Eq:LargApprox}
\E[J_K(X_K)|x_{K-1},u_{K-1}] \approx \alpha \widetilde{F}(x_{K-1},u_{K-1}) S(x_{K-1}^{(C)},x_{K-1}^{(H)}).
\end{equation}
The approximate form decomposes the energy-cost function into two product factors, depending on the state $x_{K-1}$ and action $u_{K-1}$. This simplifies the sub-optimal policy design in the sequel.

\subsubsection{Design  Sub-Optimal  Offloading-Control Policy}\label{subsubSec:subPlcy_large}
Using \eqref{Eq:LargApprox}, the cost-to-go function $J_{K-1}(x_{K-1})$ can be approximated as
%First, consider slot  $(K-1)$. The cost-to-go function is given by:
\begin{equation}\label{eq:Large_JK-1App}
  J_{K-1}(x_{K-1})\approx \min_{u_{K-1}\in\mathcal{U}(x_{K-1})}\! \l\{ R(x_{K-1},u_{K-1})+\alpha \widetilde{F}(x_{K-1},u_{K-1}) S(x_{K-1}^{(C)},x_{K-1}^{(H)}) \r\}.
\end{equation}
Denote the optimization problem in \eqref{eq:Large_JK-1App} as Problem P5. It can be easily proved that Problem P5 is convex optimization problem. Applying the Lagrange method yields the optimal solution for Problem P5, which is also the sub-optimal policy for slot $(K-1)$.
\begin{lemma}[Sub-Optimal Policy for slot $(K-1)$]\label{The:LargK-1}
\emph{By approximating the expected conditional energy-cost as in \eqref{Eq:LargApprox}, the sub-optimal offloading-control  policy for solving the optimization problem $J_{K-1}(x_{K-1})$ is:
\begin{align}
\!\!&u_{K-1}^{(\rm{lo})*}\!=\!\!\l(x_{K-1}^{(L)}\!+\!V(x_{K-1}^{(C)})x_{K-1}^{(Q)}\r)\!\!\!\l[1\!+\!\dfrac{1}{\sqrt{S(x_{K-1}^{(C)},x_{K-1}^{(H)})}} \!+\!\sqrt{\dfrac{\alpha}{\lambda}}\!\!\l(1\!-\!V(x_{K-1}^{(C)})\r)^{\frac{3}{2}}\!\!\sqrt{x_{K-1}^{(H)}}\r]^{-1},\label{eq:inf_JK-1_l*} \\
&u_{K-1}^{(\rm{of})*}=\sqrt{\dfrac{\alpha}{\lambda}x_{K-1}^{(H)}\l(1-V(x_{K-1}^{(C)})\r)} u_{K-1}^{(\rm{lo})*},\label{eq:inf_JK-1_off*}
\end{align}
where $S(x_{K-1}^{(C)},x_{K-1}^{(H)})$ and $V(x_{K-1}^{(C)})$ are defined in \eqref{eq:Large_xi_K-1} and \eqref{Eq:V}, respectively. The corresponding minimum expected user-energy consumption is approximated as
\begin{align}\label{eq:inf_J_K-1}
J_{K-1}(x_{K-1})\approx &\alpha \l(x_{K-1}^{(L)}+V(x_{K-1}^{(C)})x_{K-1}^{(Q)}\r)^3 \nn \\ &\l[1+\dfrac{1}{\sqrt{S(x_{K-1}^{(C)},x_{K-1}^{(H)})}}\!+\!\sqrt{\dfrac{\alpha}{\lambda}}\l(1-V(x_{K-1}^{(C)})\r)^{\frac{3}{2}}\sqrt{x_{K-1}^{(H)}}\r]^{-2}.
\end{align}
}
\end{lemma}
\begin{proof}
See Appendix~\ref{App:LargK-1}.
\end{proof}
\begin{remark}\emph{Computation offloading can potentially save energy when the buffer is sufficient and the channel is favorable, but it may cause more energy consumption due to demand computing. Recall that if $x_{K-1}^{(C)}=0$, $V(x_{K-1}^{(C)})$ represents the conditional probability of the helper-CPU busy state in slot $K$. It can be obtained from Lemma~\ref{The:LargK-1} that the ratio between $u_{K-1}^{(\rm{of})*}$ and $u_{K-1}^{(\rm{lo})*}$ decreases with the increase of $V(x_{K-1}^{(C)})$. This indicates that it is energy efficient to allocate more data for local computing when the helper CPU is more likely to be busy in slot $K$. In particular, the user should not offload data if $P_{00}=1$. Moreover, it can be obtained from \eqref{eq:inf_J_K-1} that $J_{K-1}(x_{K-1})$ monotonically increases with $V(x_{K-1}^{(C)})$, since a larger probability of the helper-CPU busy state tends to stop input-data offloaded to the helper and thus increases  the expected user-energy consumption.
% The reason is that greater conditional probability of idle for the helper-CPU state in slot $K$ yields smaller expected energy consumption $J_{K-1}(x_{K-1})$, which is aligned with intuition.
 }
\end{remark}

Note that using \eqref{Eq:V}, the approximated minimum expected user-energy consumption $J_{K-1}(x_{K-1})$ in \eqref{eq:inf_J_K-1} has the similar form as $J_{K}(x_{K})$ in \eqref{eq:inf_JK}.  This suggests that by performing backward induction, the sub-optimal offloading-control policy in each slot $k$ is expected to have  closed-form expressions similar to \eqref{eq:inf_JK-1_l*} and \eqref{eq:inf_JK-1_off*} by using the similar approximation techniques as for deriving \eqref{Eq:LargApprox}. To this end, following the similar procedure as for deriving Theorem~\ref{Theo:OptW/oBuf}, the sub-optimal offloading-control policy in each slot $k$ for the large-buffer case can be derived as follows.
\begin{theorem}[Sub-Optimal Policy for Multi-Slots]\label{the:large_k}
\emph{Given a large buffer ($Q\ge D$) at the helper, for $k=1,2,\dots,K-1$, the sub-optimal offloading-control policy is given as follows by approximating the expected conditional energy-cost:
\begin{align}\label{eq:inf_off_loc}
 &u_k^{(\rm{lo})*}\!=\!\l(x_k^{(L)}+V_k(x_k^{(C)})x_k^{(Q)}\r)\!\l[1+\dfrac{1}{\sqrt{\widetilde{S}_k(x_k^{(C)},x_k^{(H)})}}+\sqrt{\dfrac{\alpha}{\lambda}}\l(1-V_k(x_k^{(C)})\r)^{\frac{3}{2}}\sqrt{x_k^{(H)}}\!\r]^{-1},\\
&u_k^{(\rm{of})*}=\sqrt{\dfrac{\alpha}{\lambda}x_{k}^{(H)}\l(1-V_k(x_{k}^{(C)})\r)} u_{k}^{(\rm{lo})*},
\end{align}
where
$V_k(x_k^{(C)})\triangleq (1-x_k^{(C)})P_{00}^{K-k}$, and the function $\widetilde{S}_k(x_k^{(C)},x_k^{(H)})$ is defined as
\begin{equation}\label{eq:inf_sub_xi}
 \!\!\!\!\!\widetilde{S}_k(x_k^{(C)},x_k^{(H)})\!=\!\begin{cases} \begin{split}\!\! \sum\limits_{x_{k+1}^{(C)}\in\{0,1\}}&\sum\limits_{x_{k+1}^{(H)}\in\{g,b\}}\!\! \Pr(x_{k+1}^{(C)}|x_k^{(C)})\Pr(x_{k+1}^{(H)}|x_k^{(H)})\\
  &\!\!\!\!\!\!\!\!\!\!\!\!\!\!\!\!\!\l[1\!+\!\dfrac{1}{\sqrt{\widetilde{S}_{k+1}(x_{k+1}^{(C)},x_{k+1}^{(H)})}}+ \sqrt{\dfrac{\alpha}{\lambda}}\l(1-V_{k+1}(x_{k+1}^{(C)})\r)^{\frac{3}{2}} \!\!\sqrt{x_{k+1}^{(H)}}\r]^{-2}, k < K, \end{split}\\
  \infty , \qquad\qquad\qquad\qquad\qquad\qquad\qquad\qquad\qquad\qquad\qquad\qquad \quad \quad  k=K. \notag
  \end{cases}
\end{equation}}
\end{theorem}

Note that for slot $(K-1)$, $\widetilde{S}_{K-1}(x_{K-1}^{(C)},x_{K-1}^{(H)})$ and $V_{K-1}(x_{K-1}^{(C)})$ reduce to  $S(x_{K-1}^{(C)},x_{K-1}^{(H)})$ and $V(x_{K-1}^{(C)})$ in  \eqref{eq:Large_xi_K-1} and \eqref{Eq:V}, respectively. Moreover, one can observe that the above sub-optimal  offloading-control policy has the similar structure with the optimal policy in Theorem~\ref{Theo:OptW/oBuf}. The key difference resides in that for the large-buffer case, the offloading-control policy should account for the potential energy consumption arising from demanding computing if offloading data given helper CPU's busy state.
\begin{remark}[Approximation Accuracy and Dimensionality Reduction]
\emph{Similar to Lemma~\ref{Lem:Lb_cost}, the approximation is accurate if the user has the non-causal information of $x_k^{(C)}$ and $x_k^{(H)}$ for all $k$. In addition, note that if $x_k^{(C)}=0$, we have $V_k\l(x_k^{(C)}\r)=P_{00}^{K-k}$, which can be interpreted as the \emph{conditional helper-CPU all-busy probability}. This reflects the key fact that when the helper CPU is busy in slot $k$, the offloaded and buffered bits at the helper will consume energy from demand-computing by the local CPU, only if the helper CPU is constantly busy until the last slot. Furthermore, the proposed sub-optimal policy can achieve close-to-optimality as shown by simulation in the sequel, but it considerately reduce the computation complexity from $\mathcal{O}(16D^4\times D^2 \times K)$ to $\mathcal{O}(4K)$. }
\end{remark}
\begin{remark}[Effects of Buffer]\emph{The large helper buffer affects both the offloading and local-computing policies as follows. First, for offloading, in the case of zero buffer, the user does not offload data when the helper CPU is busy as indicated in \eqref{eq:opt_plcy}. Instead, due to the large buffer at the helper, the user can offload partial data for potentially utilizing future helper CPU resources and thus reducing transmission-energy consumption in favorable channels, which can be inferred from \eqref{eq:inf_off_loc}. Next, for local computing, with a large buffer, the user computes less data with the local CPU  than in the case of zero buffer, due to a larger offloaded data size in the busy states. These properties are confirmed by simulation results in the sequel.
}
\end{remark}

\section{Computation Offloading Control: Small Helper Buffer}\label{sec:small_bf}
In this section, the solution approach for the energy-efficient stochastic computation offloading developed in the preceding sections is extended to the case of a small buffer at the helper (i.e., $Q<D$). The closed-form expression for the optimal policy is intractable due to the buffer constraints in each slot [see \eqref{Eq:BufferCons}]. To address this issue, by exploiting the insight from policies for the earlier cases of zero and large buffers, we design a sub-optimal offloading-control policy that sheds light on the optimal policy structure and can achieve close-to-optimality as confirmed by simulation.

The derivation for the optimal policy for multi-slots encounters the following two key challenges. The first one is the difficulty for deriving a closed-form expression for the expected future energy cost, which is elaborated in Section~\ref{Sec:LargeOpt}. The second one lies in the buffer constraint (i.e., $x_k^{(Q)}+u_k^{(\rm{of})}\le Q$, if $x_k^{(C)}=0$), which bounds the amount of offloaded data according to the instantaneous buffered data size. To tackle these challenges, we propose a tractable \emph{buffer-aware candidate-policy selection} (BACS) scheme, which is defined as follows.
\begin{definition}[Buffer-Aware Candidate-Policy Selection]
\emph{The BACS scheme chooses one of the two candidate policies as the sub-optimal computation offloading policy, namely, \emph{truncated large-buffer based policy} (TLBP) and \emph{zero-buffer based policy} (ZBP) defined in the sequel, depending on the helper buffer size $Q$. To be specific, the BACS scheme selects the TLBP or ZBP policy, if the buffer size $Q$ is above or below a \emph{switching threshold} $Q_{\rm{th}}$ that is specified in the sequel.
}
\end{definition}

The key idea of the proposed BACS is to exploit the useful structures in the developed policies for the earlier cases of zero buffer and large buffer, and use them to design the sub-optimal policy for the case of small helper buffer. Intuitively, if the buffer size is relatively large (i.e., $Q\to D^{-}$), modifying the sub-optimal policy for the large buffer (see Section~\ref{Sec:LargeBuffer}) is expected to achieve desirable energy-savings performance. This modified policy, however, is unsuitable for the case of a relatively small buffer (i.e., $Q\to 0^{+}$), as it is based on the assumption of a large buffer. This suggests to choose another candidate policy that is modified from the optimal policy for the zero buffer. In the following, we elaborate both the TLBP and ZBP policies, and characterize the switching threshold.

%BCS indicates that in the first slot, if $Q<Q_{th}$, the sub-optimal co-computing policy is equivalent to OPWB. Otherwise, the co-computing policy follows TPLB. OPWB and TPLB are presented as follows.
\subsubsection{Truncated Large-Buffer based Policy}
This candidate policy is designed by applying the \emph{relaxation-and-truncation} (RT) approach for the solution approach of the large-buffer counterpart, accounting for the limited buffer size. To this end, in each slot, we first relax the buffer constraint and thus the problem reduces to the large-buffer counterpart. The corresponding sub-optimal policy can be derived in Theorem~\ref{the:large_k}, denoted by $\l\{\widehat{u}_k^{(\rm{lo})*}, \widehat{u}_k^{(\rm{of})*}\r\}$. Next, to account for the small buffer size, the control policy is further given by
\begin{equation}\label{eq:plcy_small}
u_k^{(\rm{lo})*}=\widehat{u}_k^{(\rm{lo})*},\quad\text{and}\quad
u_k^{(\rm{of})*}=\begin{cases} \widehat{u}_k^{(\rm{of})*}, & x_k^{(C)}=1, \\
\Big\langle \widehat{u}_k^{(\rm{of})*} \Big\rangle_{0}^{Q-x_k^{(Q)}}, \quad &x_k^{(C)}=0,
\end{cases}
\end{equation}
where $\langle\cdot\rangle_{z_1}^{z_2}$ denotes the truncation below $z_1$ and above $z_2$. Note that the offloaded data only needs to be truncated if the helper CPU is busy.
\subsubsection{Zero-Buffer based Policy}
Recall that the optimal policy for the case of zero buffer (see Section~\ref{Sec:ContWoBuf}) does not offload data when the helper CPU is busy. This key fact allows directly using the solution approach for the zero-buffer counterpart  to the current case of a small buffer.
\subsubsection{Switching Threshold} Let $G_1(Q)$ and $G_2(Q)$ denote the expected user-energy consumption with a  $K$-slot deadline  by using TLBP and ZBP, respectively, which can be numerically computed. First, the  existence of the switching threshold is established in the following lemma.
\begin{lemma}[Existence of the Switching Threshold]\label{lem:exist_Q_th}\emph{There exists a switching threshold  $Q_{\rm{th}}$ in the range of  $0\le Q_{\rm{th}}\le D$, which satisfies $G_1(Q_{\rm{th}})= G_2(Q_{\rm{th}})$.
}
\end{lemma}
\begin{proof}
See Appendix~\ref{App:exist_Q_th}.
\end{proof}

In addition, though the expected user-energy consumption using TLBP, $G_1(Q)$, is intractable due to the RT procedure, it can be evaluated by  simulation that $G_1(Q)$ is \emph{non-increasing} with the buffer size $Q$. This is expected since a larger buffer size allows more data to be offloaded, if offloading consumes less energy than local computing.  The effects of the computing deadline on the switching threshold are evaluated by simulation in the sequel. The switching threshold can be efficiently computed by a proposed bisection-search procedure as detailed in Algorithm~\ref{alg:Q_th}.
\begin{algorithm}[t!]
\caption{Proposed Bisection-Search Algorithm for the Switching Threshold} \label{alg:Q_th} \centering
\begin{algorithmic}
\STATE \textbf{Step 1}: Calculate the expected user-energy consumption  over $10^4$ realizations for ZBP using the results in Theorem~\ref{Theo:OptW/oBuf}, denoted by $\bar{E}_Z$.
\STATE \textbf{Step 2} [Bisection-search for $Q_{\rm{th}}$]: Initialize $Q_{\ell}=0$, $Q_{h}=D/2$, and $\varepsilon >0$. \emph{Repeat}
\STATE (1) Set $Q_m=(Q_{\ell}+Q_{h})/2$, and calculate the expected user-energy consumption over $10^4$ realizations for TLBP, denoted as  $\bar{E}_T$.
% based on TLBP.
\STATE (2) If $\bar{E}_T<\bar{E}_Z$, let $Q_{h} = Q_m$; otherwise, $Q_{\ell}=Q_m$.
\STATE \emph{Until} $|\bar{E}_T-\bar{E}_Z| \le \varepsilon$. Return $Q_{\rm{th}}=Q_m$.
\end{algorithmic}
\end{algorithm}

\section{Simulation Results and Discussions}\label{sec: simulation results}
In this section, simulation results are presented for evaluating the performance of the proposed stochastic computation offloading policies. The parameters are set as follows unless stated otherwise. The user is required to compute $3000$-bit given a $5$-slot deadline, each slot with a time  duration $t_0=100$ ms. For local computing, $\gamma=10^{-28}$ and the required CPU cycles for computing 1-bit data is $w=10^{5}$ cycle/bit\cite{you2016energy,you2017energy}. For offloading, the channel gain $H_k$ between the user and the helper follows a stationary Markov chain with the transition probabilities set as $P_{gg}=0.8$ and $P_{bb}=0.7$. The channel gains of the two channel states are set as $g=10^{-3}$ and $b=10^{-5}$, respectively. The energy coefficient is $\lambda=10^{-15}$. The helper-CPU state  follows another Markov chain with $P_{11}=0.8$ and $P_{00}=0.7$. For the case of small buffer, the helper buffer size is $Q=300$ bits. For performance comparison, the optimal policy is numerically computed by DP. Furthermore, a baseline algorithm is considered, called \emph{equal-allocation} policy, which first equally allocates the total data into $K$ slots and then optimizes the locally-computed and offloaded data sizes in each slot.

\subsection{Close-to-Optimality}
\begin{figure}[t!]
  \centering
  \includegraphics[width=7.3cm]{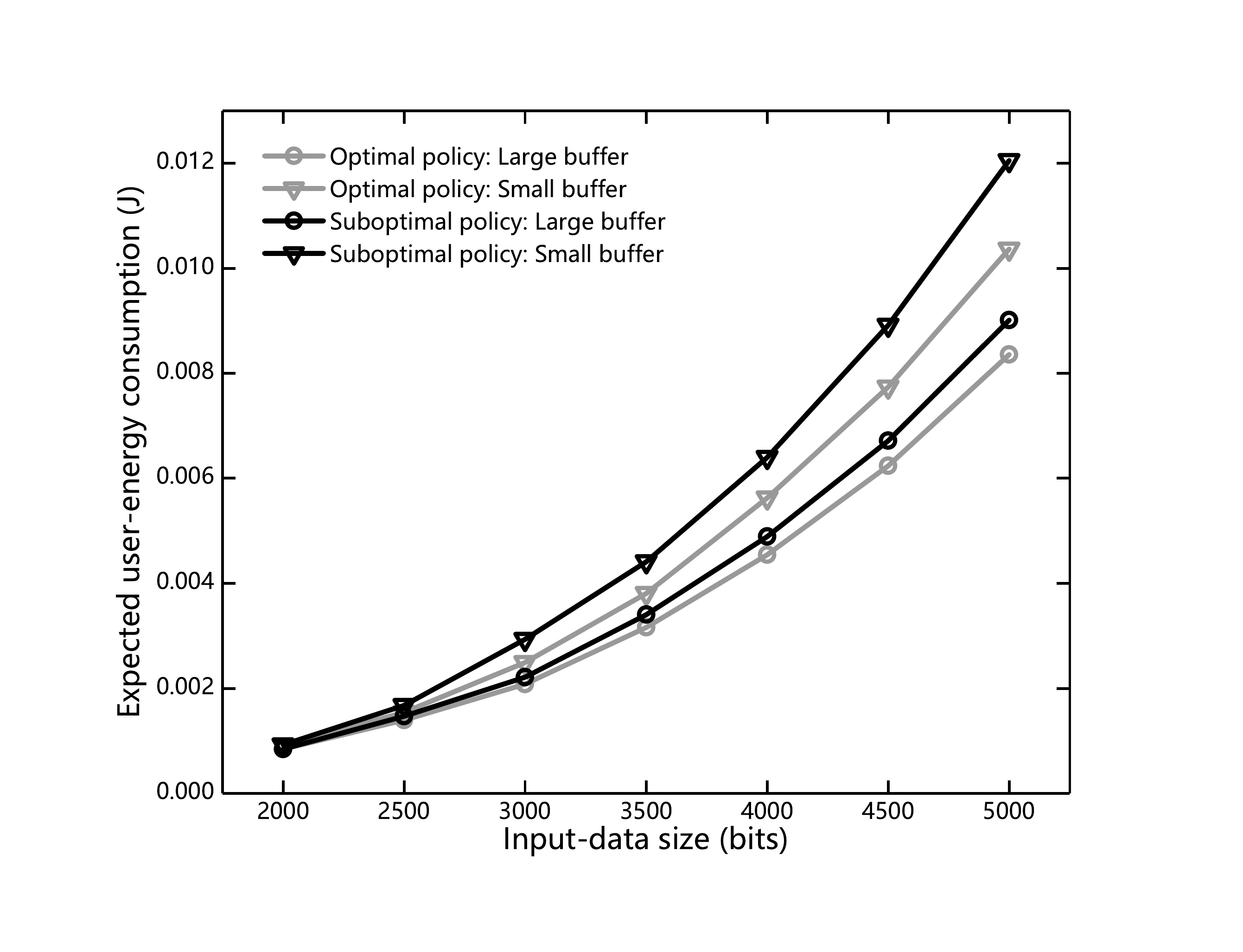}\\
  \caption{Close-to-optimality of proposed policies.}\label{Fig:Close_opt}
\end{figure}
Fig.~\ref{Fig:Close_opt} shows close-to-optimal performance of the proposed sub-optimal policies. It is observed that for the case of a large buffer, the expected user-energy consumption of the sub-optimal policy is close to that of the optimal one, especially when the data size is small. Moreover, even in the large-data regime e.g., $D=5000$ bits, the sub-optimal policy only results in additional $7\%$ of the user-energy consumption. Next, for the case of small buffer, the proposed sub-optimal policy also achieves close-to-optimal performance and has slightly larger user-energy consumption in the large-data regime.
\subsection{Effects of Parameters}
\begin{figure}[t!]
  \centering
  \subfigure[Effects of data size]{
  \begin{minipage}{7.5cm}
  \centering
  \includegraphics[width=7.3cm]{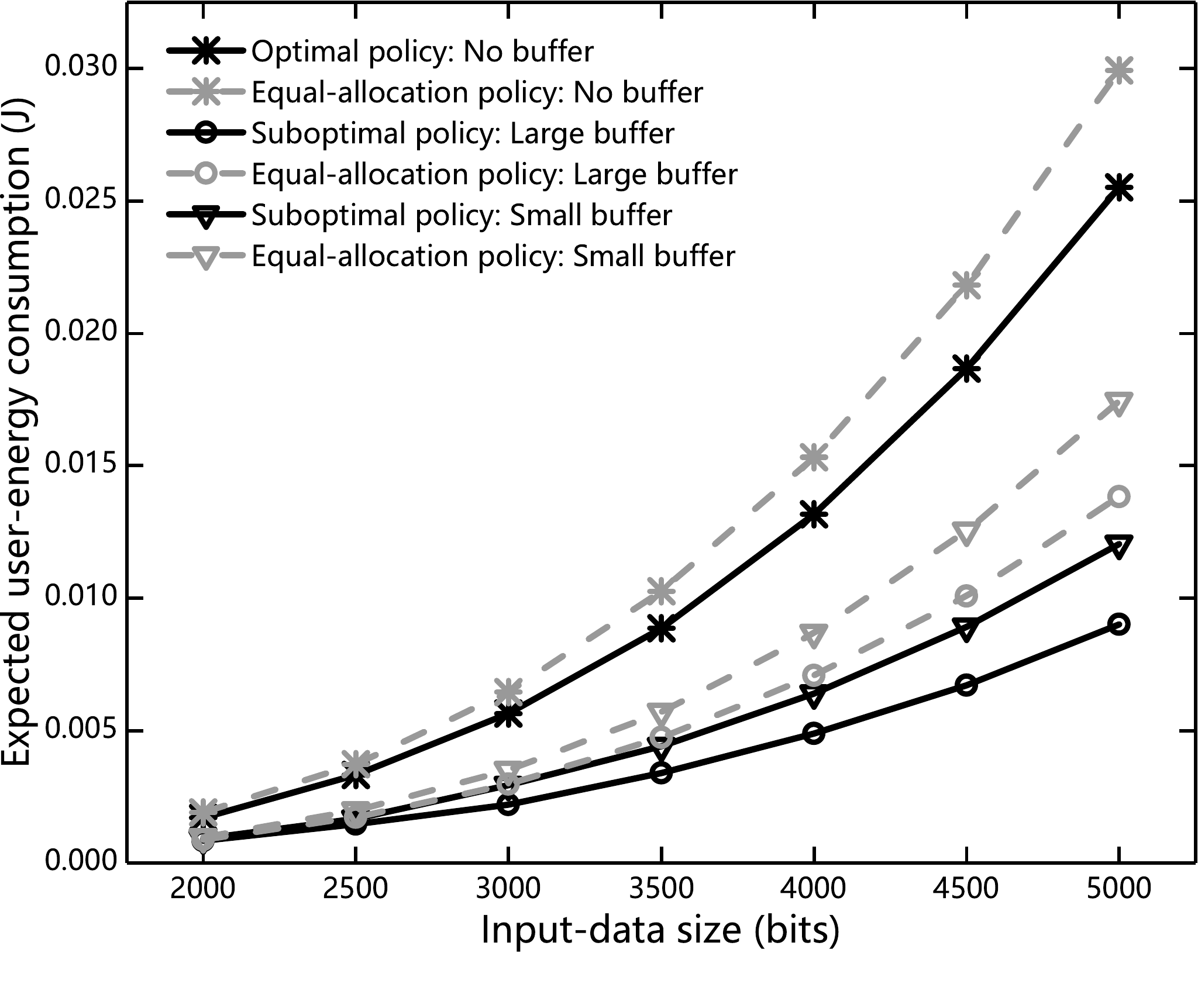}\label{Fig:Eff_datasize}
  \end{minipage}
  }
  \subfigure[Effects of computing deadline]{
  \begin{minipage}{7.5cm}
  \centering
  \includegraphics[width=7.3cm]{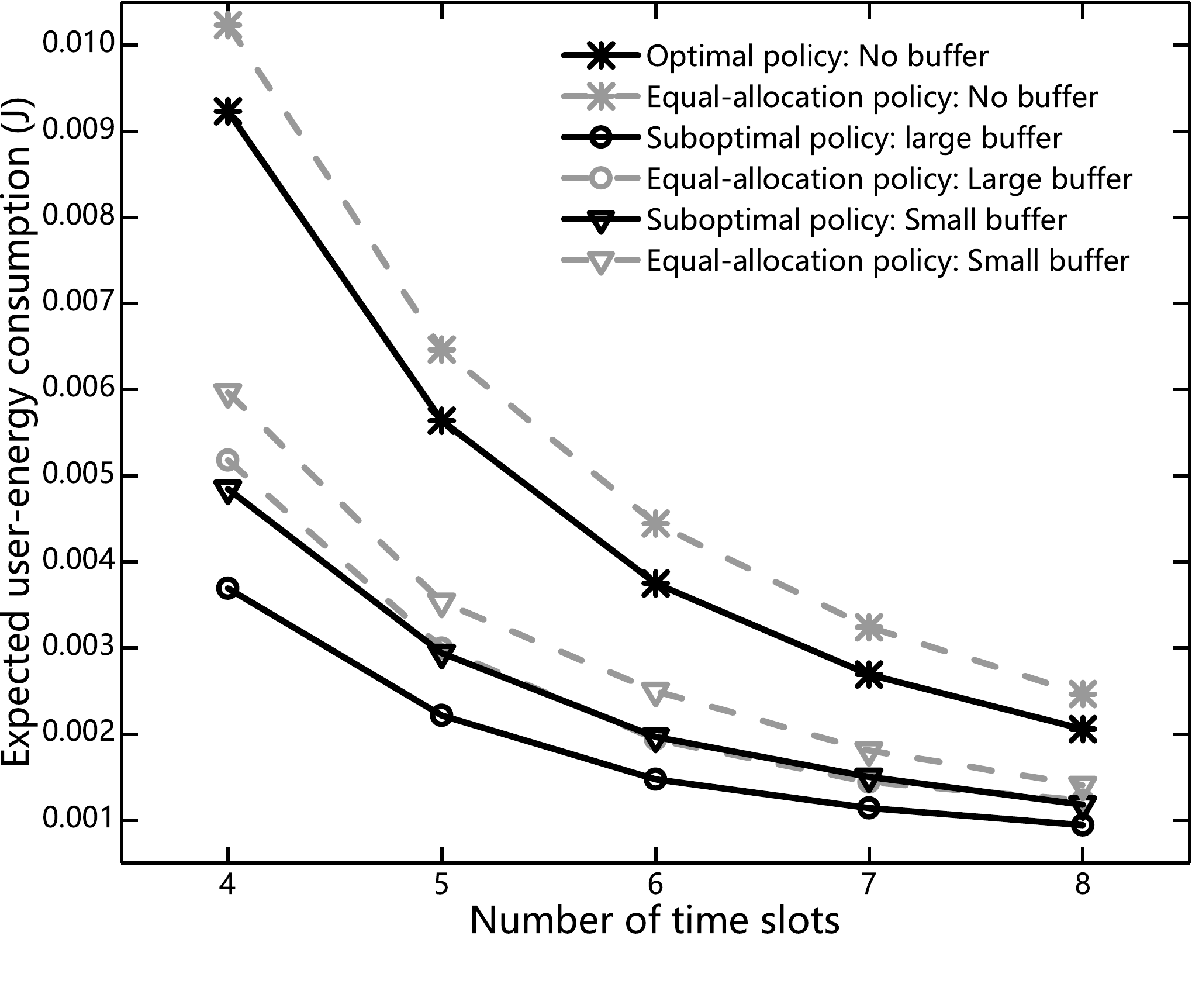}\label{Fig:Effects_deadline}
  \end{minipage}
  }
  \subfigure[Effects of the helper-CPU idle probability]{
  \begin{minipage}{7.5cm}
  \centering
  \includegraphics[width=7.3cm]{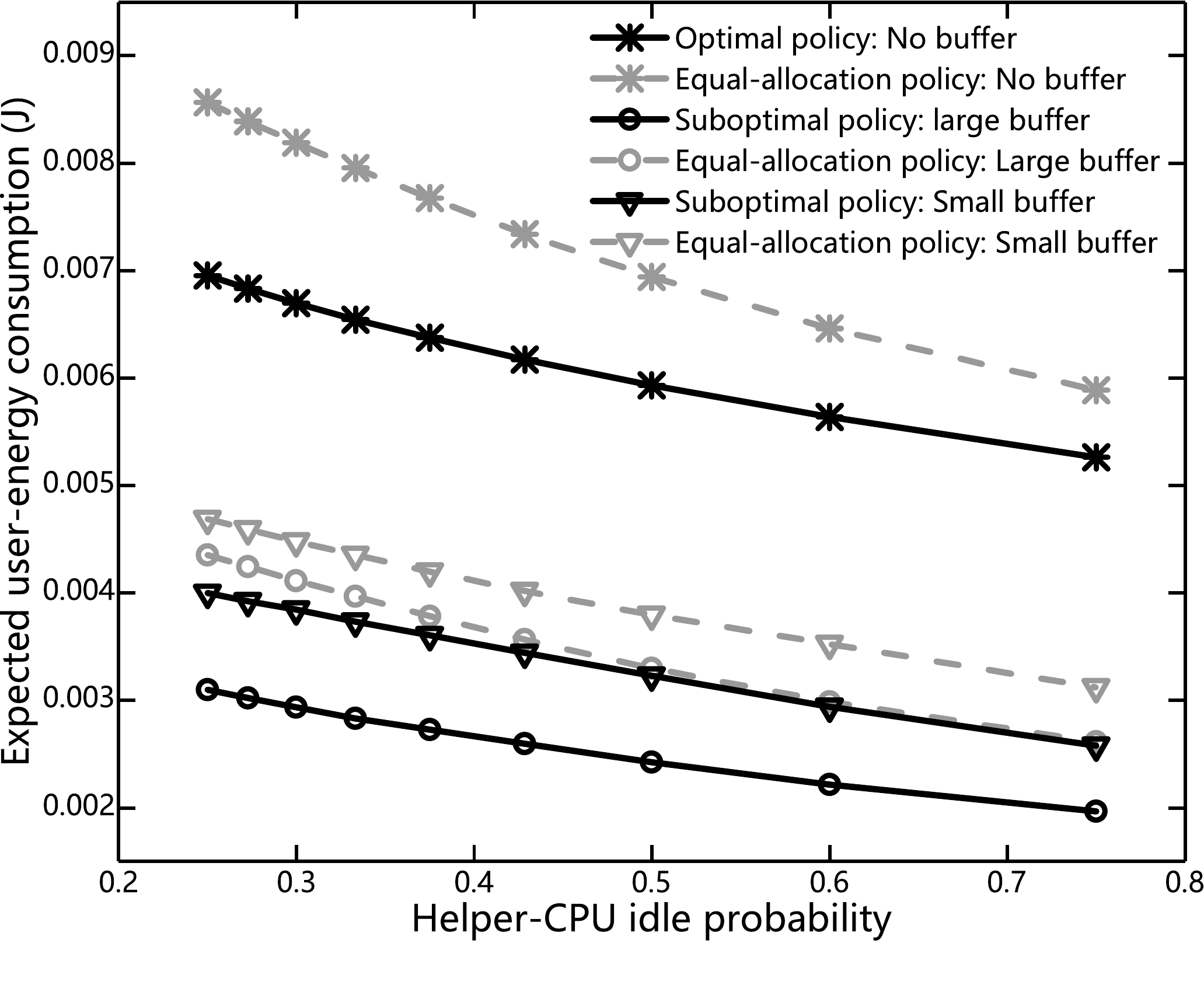}\label{Fig:Effects_idling}
  \end{minipage}
  }
  \subfigure[Effects of buffer size]{
  \begin{minipage}{7.5cm}
  \centering
  \includegraphics[width=7.3cm]{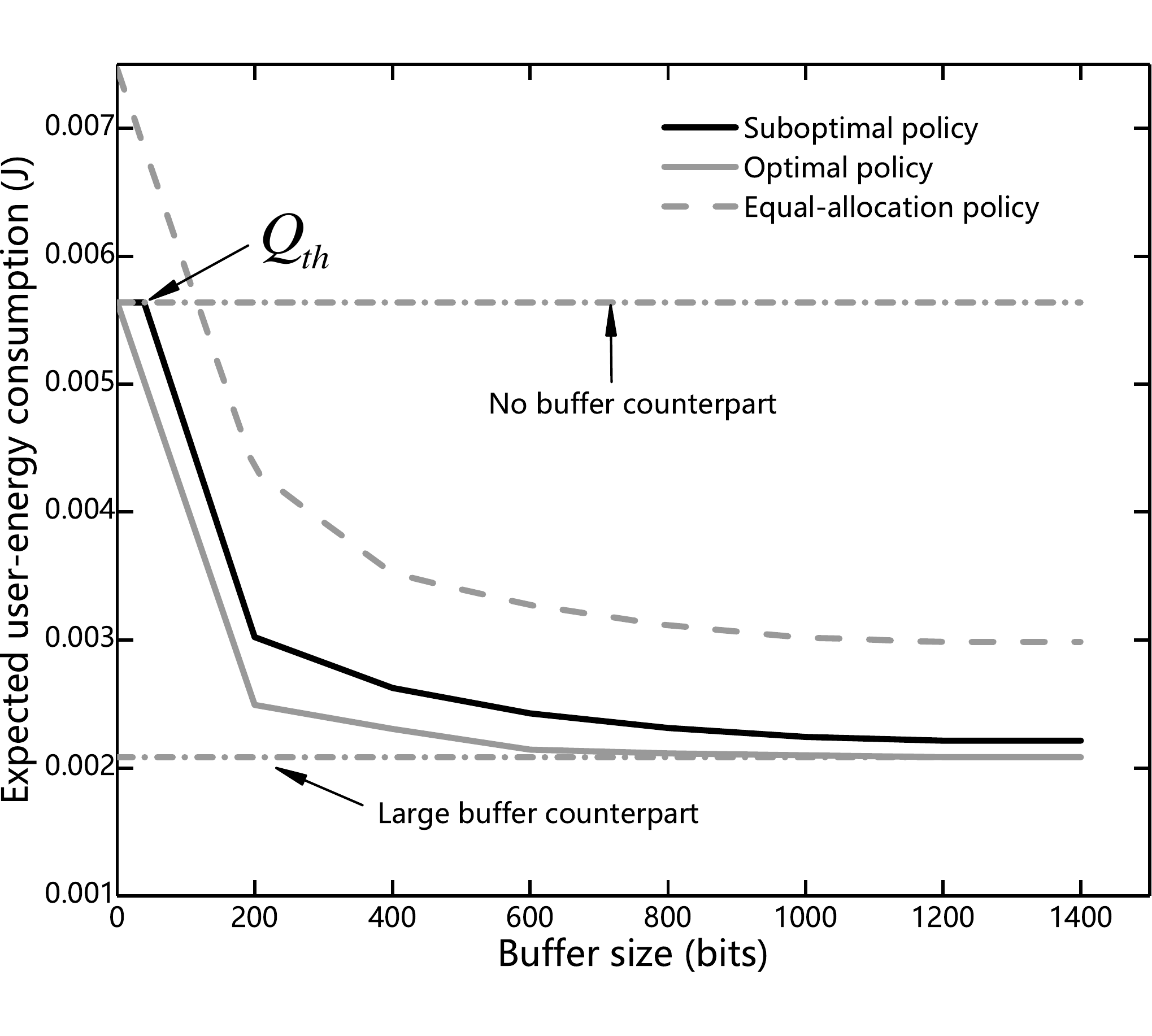}\label{Fig:Effects_buffersize}
  \end{minipage}
  }
  \caption{The effects of parameters on the expected user-energy consumption for the three cases: (a) input-data size with $K=5$; (b) computing deadline with $D=3000$ bits; (c) helper-CPU idle probability with $D=3000$ bits and $K=5$; (d) buffer size with $D=3000$ bits and $K=5$. }\label{Fig:Effects_para}
\end{figure}
The curves of the expected user-energy consumption versus the input-data size are plotted in Fig.~\ref{Fig:Eff_datasize}. Several observations are made as follows. First, as the input-data size increases, the expected user-energy consumption of each policy grows at an increasing rate. Next, for the case of zero buffer, the optimal policy has less user-energy consumption than the equal-allocation policy, especially in the large-data regime. Similarly, the proposed sub-optimal policies for the case of large/small buffer achieve significant performance gains, both of them achieving about $30\%$ reduction of energy consumption compared with the equal-allocation policies at $D\approx5000$ bits.  Moreover, for the proposed optimal/sub-optimal policies, the policy of the large-buffer case consumes less user-energy consumption than the case of small buffer; and the case of zero buffer consumes the largest user-energy consumption. This shows that increasing the buffer size can potentially reduce the user-energy consumption, as it allows more data to be offloaded to the helper for energy savings.

Fig.~\ref{Fig:Effects_deadline} depicts the curves of expected user-energy consumption versus the computing deadlines. Observe that the expected user-energy consumption of each policy is monotonically decreasing with the increasing of number of  slots, which is aligned with intuition. In particular, the decreasing rate is larger at a shorter deadline (i.e., smaller number of total slots), indicating that it is highly energy-efficient to extend the deadline for energy savings when the current deadline requirement is stringent. In addition, for the case of no buffer, the optimal policy has almost constant performance gain compared with the equal-allocation policy, while the gain deceases with the number of slots for the cases of large and small buffer. Moreover, the proposed policy for the case of large buffer has much smaller user-energy consumption than the counterparts with a small buffer and no buffer.

The effects of the helper-CPU idle probability on the expected user-energy consumption are evaluated in Fig.~\ref{Fig:Effects_idling}. It can be observed that as the helper-CPU idle probability increases, the user-energy consumptions of the proposed policies are almost \emph{linearly} decreasing. The reason is that, with a higher idle probability, the helper CPU is more likely at the idle state in the whole slots, allowing the user to offload more data to the helper for energy savings. Again, the performances of proposed sub-optimal policies significantly outperform the corresponding equal-allocation policies. Other observations are similar to those of Fig.~\ref{Fig:Effects_deadline}.

The curves of expected user-energy consumption versus the buffer size are shown in Fig.~\ref{Fig:Effects_buffersize}. It is observed that enlarging the buffer size can considerably help reduce the expected user-energy consumption of the proposed sub-optimal policy when the buffer size is small. However, after the buffer size exceeds a threshold (about $1000$ bits), the performance cannot be further improved and converges to that of the large-buffer counterpart. The reason is that the buffer size is no longer the bottleneck for offloading data, when the buffer size is sufficiently large. Next, the switching threshold $Q_{\rm{th}}$ is quite small (about $40$ bits), which means that for the helper with a moderate buffer size, the ZBP policy outperforms the TLBP policy. In addition, the proposed sub-optimal policy has close-to-optimal performance, especially in the large-buffer size regime.
\subsection{Buffer Gain}
\begin{figure}[t!]
  \centering
  \subfigure[Effects of buffer size]{
  \begin{minipage}{7.5cm}
  \centering
  \includegraphics[width=7.3cm]{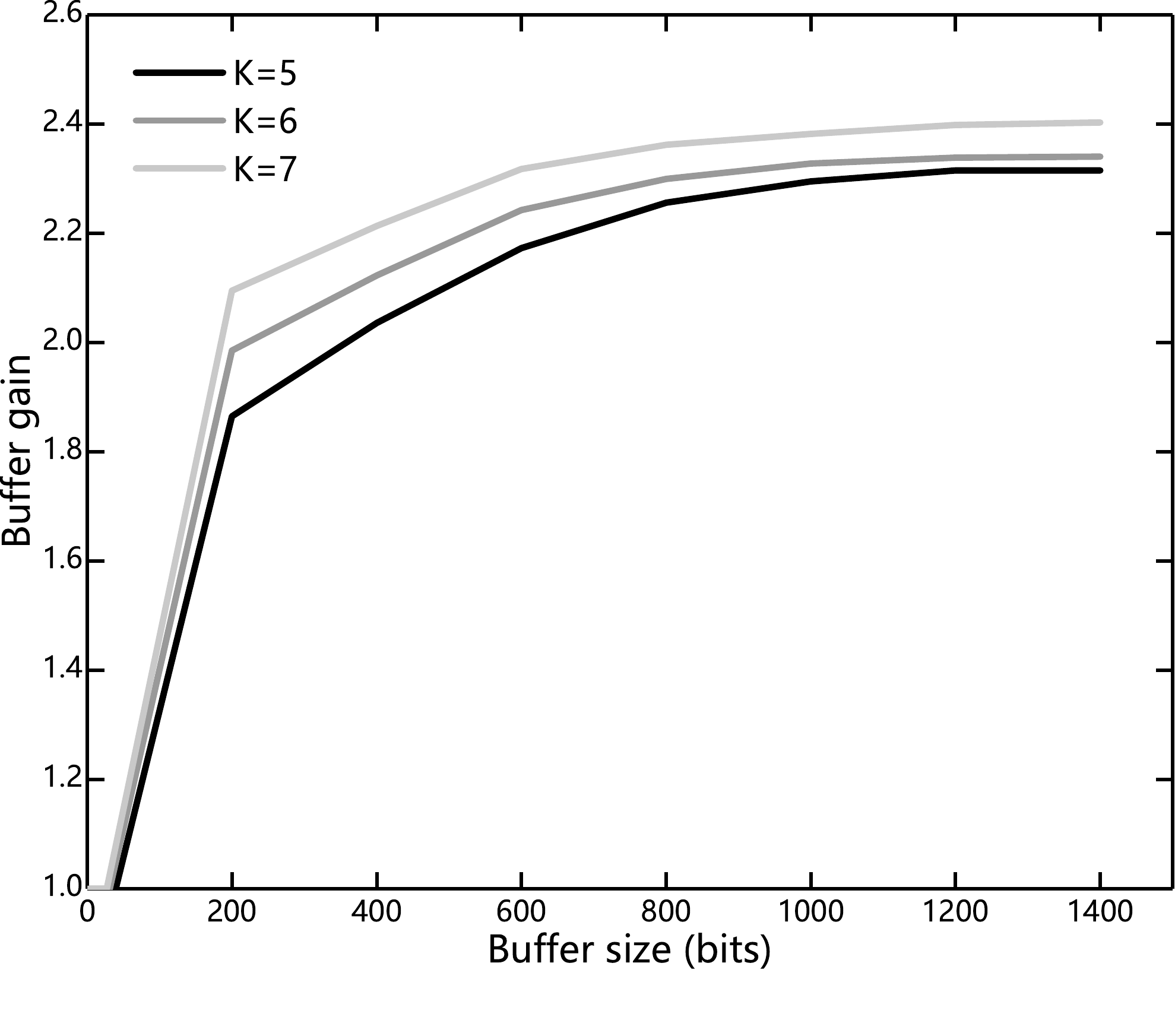}\label{Fig:bf_gain}
  \end{minipage}
  }
  \subfigure[Effects of the helper-CPU idle probability]{
  \begin{minipage}{7.5cm}
  \centering
  \includegraphics[width=7.3cm]{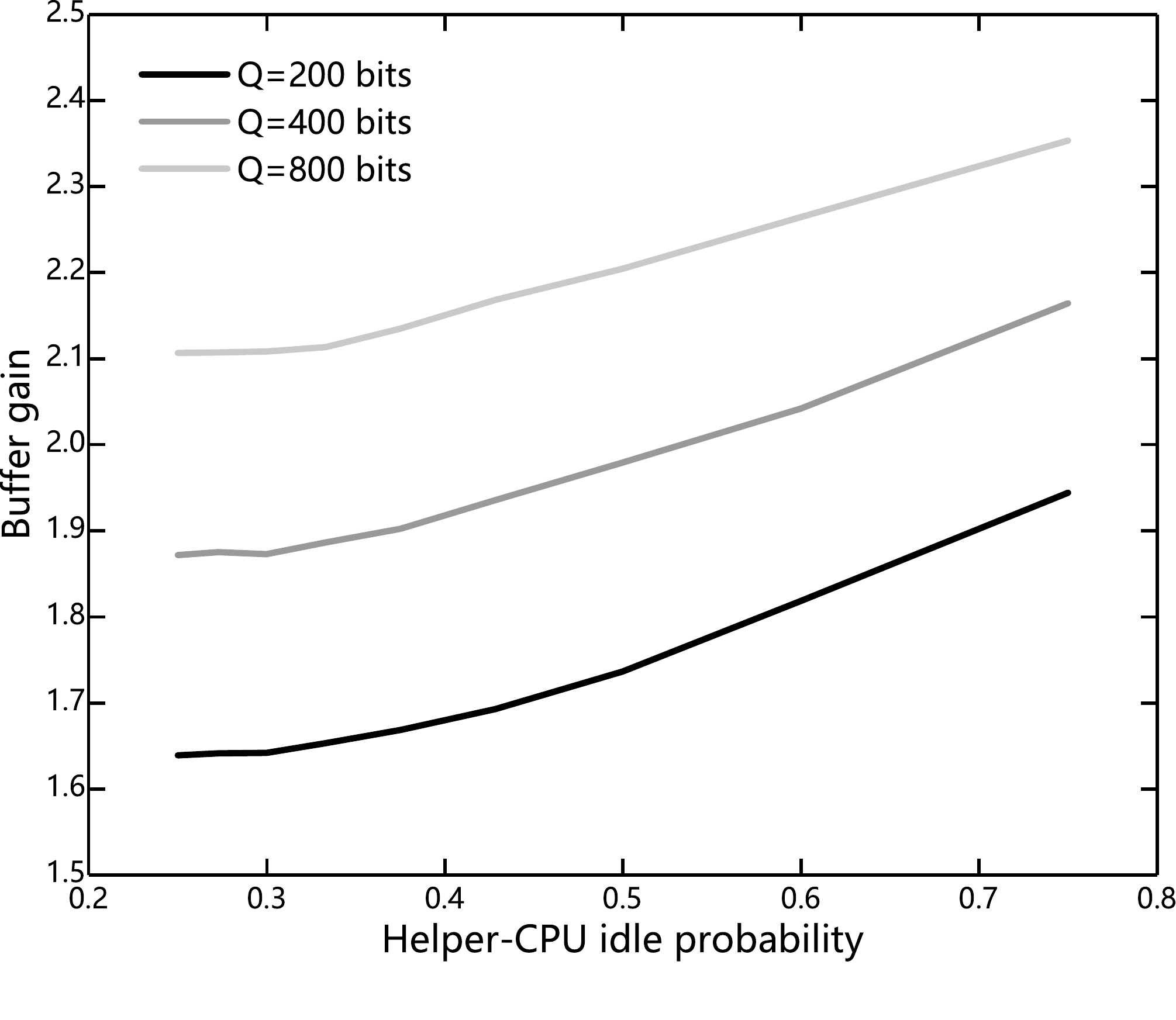}\label{Fig:buffer_vs_P11}
  \end{minipage}
  }
  \caption{The effects of parameters on the buffer gain: (a) buffer size with $D=3000$ bits; (b) helper-CPU idle probability with $D=3000$ bits and $K=5$.}\label{Fig:Effects_buffer}
\end{figure}
To quantify the amount of energy-consumption reduction obtained from the helper buffer, we define a performance metric, called \emph{buffer gain}, as the expected user-energy consumption ratio between the computation offloading policies without and with a buffer $Q$.

Fig.~\ref{Fig:bf_gain} shows the effects of buffer size on the buffer gain with different computing deadlines. It is observed that the buffer gains for different deadlines see a fast increase in the small buffer-size regime and then converge gradually to a constant corresponding to the case of a large buffer. \emph{It implies that the maximum buffer gain is finite and independent of the buffer size.} Moreover, the proposed policy with a longer computing deadline has a relatively larger buffer gain.

The curves of buffer gain versus the helper-CPU idle probability are depicted in Fig.~\ref{Fig:buffer_vs_P11}. Observe that given a fixed buffer size, as the helper-CPU idle probability increases, the buffer gain firstly slowly increases and then sees a fast and \emph{almost-linear} increase when the helper-CPU idle probability is larger than $0.3$. The reason is that, with a higher idle probability, the buffered data is more likely to be computed by the helper. On the other hand, given a fixed helper-CPU idle probability, the larger buffer size enjoys a larger buffer gain, which is expected.

\section{Concluding Remarks}\label{sec: Concluding Remarks}
This paper investigates the stochastic computation offloading to a helper with a dynamically loaded CPU. Assuming that the user only has statistical information of the helper CPU and channel states, one MDP optimization problem is formulated to minimize the expected user-energy consumption by controlling the offloaded and locally-computed data sizes in different slots. We first characterize the optimal policy structure for the case of zero buffer at the helper. Then for the cases of large and small helper buffer, we propose effective sub-optimal computation offloading policies by using approximation techniques and candidate-policy selection scheme, respectively, which are shown to have close-to-optimal performance in simulation. This work can be extended in several interesting directions. First, the current work focusing on one-shot task arrival can be generalized to the case of bursty data arrivals. Second, P2P computation offloading can be integrated with advanced computing techniques such as parallel computing to achieve higher energy savings gain. Last, it is interesting to extend the current helper-assisted computation offloading system to a more complex one, e.g., finite computation capacity at helpers and collaborative offloading among mobiles.

\appendix

\subsection{Proof of Lemma~\ref{Lem:ConK}} \label{app:conK}
First, if $x_K^{(C)}=0$, computation offloading is not allowed and thus the user-energy consumption is directly given by $J_K(x_K)=\alpha (x_K^{(L)})^3$. Next, if $x_K^{(C)}=1$, taking the first-order derivative of the objective function in  Problem P3 and setting it to zero yields the optimal local computed data size $u_K^{(\rm{lo})*}=\left( 1+\sqrt{\alpha x_K^{(H)} /\lambda}\right )^{-1}x_K^{(L)}$. The corresponding minimum energy consumption  $J_K(x_K)=\alpha (x_K^{(L)})^3 \left( 1+ \sqrt{\alpha x_K^{(H)} /\lambda}\right )^{-2}$. Last, combining the results for $x_K^{(C)}=0$ and $x_K^{(C)}=1$ together leads to the desired results.

\subsection{Proof of Theorem~\ref{Theo:OptW/oBuf}}\label{app:opt_data_scheduling}
We first derive the minimum expected user-energy from  slot $(K-1)$ to $K$ for both the cases of  $x_{K-1}^{(C)}=1$ and  $x_{K-1}^{(C)}=0$ as follows. If $x_{K-1}^{(C)}=0$, the input-data cannot be offloaded, i.e., $u_{K-1}^{(\rm{of})}=0$. Then the optimization problem reduces to
\begin{align}
\!\!\!\!J_{K-1}(x_{K-1})&= \min \limits_{0\le u_{K-1}^{(\rm{lo})}\le x_{K-1}^{(L)}} \left \{\alpha (u_{K-1}^{(\rm{lo})})^3 + \E \l[J_K(X_K)\Big|x_{K-1},u_{K-1}^{(\rm{lo})}\r] \right\} \nn. \\
    &\overset{(\text{i})}{=}\min \limits_{0\le u_{K-1}^{(\rm{lo})}\le x_{K-1}^{(L)}} \left\{\alpha (u_{K-1}^{(\rm{lo})})^3 + \alpha (x_{K-1}^{(L)}-u_{K-1}^{(\rm{lo})})^3 \times S_{K-1}(x_{K-1}^{(C)},x_{K-1}^{(H)})\right\}, \label{eq:J_K-1_b}
\end{align}
where (i) is according to Lemma~\ref{Lem:ConK} and $$S_{K-1}(x_{K-1}^{(C)},x_{K-1}^{(H)})=\sum\limits_{x_K^{(C)}\in\{0,1\}}\sum\limits_{x_K^{(H)}\in\{g,b\}}\Pr(x_K^{(C)}|x_{K-1}^{(C)}) \Pr(x_K^{(H)}|x_{K-1}^{(H)})\l(1+\sqrt{\dfrac{\alpha}{\lambda}}x_K^{(C)}\sqrt{x_K^{(H)}}\r)^{-2}.$$ Taking the first-order derivative of the objective function in \eqref{eq:J_K-1_b} and setting it to zero, we have
\begin{equation}\label{eq:2JK-1_b}
  u_{K-1}^{(\rm{lo})*} = \l(1+ \frac{1}{\sqrt{S_{K-1}(x_{K-1}^{(C)},x_{K-1}^{(H)})}}\r)^{-1}x_{K-1}^{(L)}.
\end{equation}
Substituting (\ref{eq:2JK-1_b}) into (\ref{eq:J_K-1_b}) yields:
\begin{equation}\label{eq:2J_K-1_b}
  J_{K-1}(x_{K-1}) = \alpha (x_{K-1}^{(L)})^3 \l(1+ \frac{1}{\sqrt{S_{K-1}(x_{K-1}^{(C)},x_{K-1}^{(H)})}}\r)^{-2}, ~~\text{if}~~ x_{K-1}^{(C)}=0.
\end{equation}
If $x_{K-1}^{(C)}=1$, both the locally-computed and offloaded data sizes should be optimized. The optimization problem reduces to
\begin{equation*}\label{eq:no_JK-1_idle}
 \!\!\!\!J_{K-1}(x_{K-1})= \!\!\!\!\!\!\min\limits_{u_{K-1}\in\mathcal{U}(x_{K-1})}\!\! \left\{\alpha (u_{K-1}^{(\rm{lo})})^3 + \lambda\frac{(u_{K-1}^{(\rm{of})})^3}{x_{K-1}^{(H)}} + \alpha \l(x_{K-1}^{(L)}-u_{K-1}^{(\rm{lo})}-u_{K-1}^{(\rm{of})}\r)^3 S_{K-1}(x_{K-1}^{(C)},x_{K-1}^{(H)})\right\}.
\end{equation*}
Using the Lagrange method leads to the following optimal local computing and offloading policy:
\begin{equation}\label{eq:No_K-1_plcy_i}
 \!u_{K-1}^{(\rm{lo})*}=\l(1+ \frac{1}{\sqrt{S_{K-1}(x_{K-1}^{(C)},x_{K-1}^{(H)})}}+\sqrt{\frac{\alpha}{\lambda} x_{K-1}^{(H)}}\r)^{-1}x_{K-1}^{(L)},~~\text{and}~~u_{K-1}^{(\rm{of})*}=\sqrt{\frac{\alpha}{\lambda} x_{K-1}^{(H)}}u_{K-1}^{(\rm{lo})*}.\!\!
\end{equation}
The corresponding minimum expected user-energy consumption is:
\begin{equation}\label{eq:JK-1_i}
  J_{K-1}(x_{K-1})=\alpha (x_{K-1}^{(L)})^3 \l(1+ \frac{1}{\sqrt{S_{K-1}(x_{K-1}^{(C)},x_{K-1}^{(H)})}}+\sqrt{\frac{\alpha}{\lambda} x_{K-1}^{(H)}}\r)^{-2}, ~\text{if}~x_{K-1}^{(C)}=1.
\end{equation}
Then combining \eqref{eq:2JK-1_b} and \eqref{eq:No_K-1_plcy_i} leads to the optimal offloading-control policy in  slot $(K-1)$:
\begin{align}\label{eq:No_K-1_plcy}
 & u_{K-1}^{(\rm{lo})*}=\l(1+ \frac{1}{\sqrt{S_{K-1}(x_{K-1}^{(C)},x_{K-1}^{(H)})}}+\sqrt{\dfrac{\alpha}{\lambda}}x_{K-1}^{(C)} \sqrt{x_{K-1}^{(H)}}\r)^{-1} x_{K-1}^{(L)},\\
  &u_{K-1}^{(\rm{of})*}=\sqrt{\dfrac{\alpha}{\lambda}}x_{K-1}^{(C)} \sqrt{x_{K-1}^{(H)}}u_{K-1}^{(\rm{lo})*}.
\end{align}
The minimum expected user-energy consumption for slot $(K-1)$ to $K$ can be derived as below by combining \eqref{eq:2J_K-1_b} and \eqref{eq:JK-1_i}:
\begin{align}\label{eq:JK-1}
  J_{K-1}(x_{K-1})=\alpha (x_{K-1}^{(L)})^3 \l(1+ \frac{1}{\sqrt{S_{K-1}(x_{K-1}^{(C)},x_{K-1}^{(H)})}} +\sqrt{\dfrac{\alpha}{\lambda}}x_{K-1}^{(C)} \sqrt{x_{K-1}^{(H)}}\r)^{-2}.
\end{align}
Observe that the expected user-energy consumption from slot ($K-1$) and $K$, given in \eqref{eq:JK-1}, has the similar form with that in slot $K$ as given in \eqref{Eq:ConKEgy}.  Therefore, using the backward induction and the similar procedure as for deriving \eqref{eq:JK-1}, we can derive the optimal computation offloading policy in each slot $k$ and the minimum expected user-energy consumption, as expressed in Theorem~\ref{Theo:OptW/oBuf}, completing the proof.

\subsection{Proof of Lemma~\ref{Lem:Lb_cost}}\label{App:Lb_cost}
We prove this lemma by deriving the minimum expected user-energy for both the cases of $x_{K-1}^{(C)}=1$ and $x_{K-1}^{(C)}=0$ respectively, and then combining the results together. First, if $x_{K-1}^{(C)}=1$, we have  $F(x_{K-1}, u_{K-1}, x_{K}^{(C)})=(x_{K-1}^{(L)}-u_{K-1}^{(\rm{lo})}-u_{K-1}^{(\rm{of})})^3$ where $F(x_{K-1}, u_{K-1}, x_{K}^{(C)})$ is defined in \eqref{EqF}, and thus
\begin{equation}\label{eq:inf_barJK_i}
\E\l[ J_K(X_K)\big|x_{K-1},u_{K-1} \r] = \alpha (x_{K-1}^{(L)}-u_{K-1}^{(\rm{lo})}-u_{K-1}^{(\rm{of})})^3 \times S(x_{K-1}^{(C)},x_{K-1}^{(H)}), ~~\text{if}~ x_{K-1}^{(C)}=1,
\end{equation}
where $S(x_{K-1}^{(C)},x_{K-1}^{(H)})$ is defined in \eqref{eq:Large_xi_K-1}.
Next, if $x_{K-1}^{(C)}=0$, the newly-offloaded data in slot $(K-1)$ is stored at the buffer and thus $$F(x_{K-1}, u_{K-1}, x_{K}^{(C)})\!=\!\l[(x_{K-1}^{(L)}-u_{K-1}^{(\rm{lo})}-u_{K-1}^{(\rm{of})})+ (1-x_K^{(C)})(x_{K-1}^{(Q)}+u_{K-1}^{(\rm{of})})\r]^3.$$ Then it follows
\begin{align}\label{eq:largeK_approx}
\!\!\!&\E\l[ J_K(X_K)\big|x_{K-1},u_{K-1} \r] \nn \\
  & \overset{(\text{i})}{\ge} \alpha \E\l[F(x_{K-1}, u_{K-1}, X_{K}^{(C)}) \bigg|x_{K-1},u_{K-1} \r] \times \E\l[\l(1+\sqrt{\dfrac{\alpha}{\lambda}}X_{K}^{(C)}\sqrt{X_{K}^{(H)}}\r)^{-2} \bigg| x_{K-1}\r] \nn \\
  & \overset{(\text{ii})}{\ge} \alpha\l(\E\l[(x_{K-1}^{(L)}-u_{K-1}^{(\rm{lo})}-u_{K-1}^{(\rm{lo})})+ (1-X_K^{(C)})(x_{K-1}^{(Q)}+u_{K-1}^{(\rm{of})}) \bigg|x_{K-1},u_{K-1} \r]\r)^3 \times S(x_{K-1}^{(C)},x_{K-1}^{(H)}) \nn  \\
  &\overset{(\text{iii})}{=} \alpha \l(x_{K-1}^{(L)}-u_{K-1}^{(\rm{lo})}-u_{K-1}^{(\rm{lo})})+P_{00}(x_{K-1}^{(Q)}+u_{K-1}^{(\rm{of})})\r)^3 \times S(x_{K-1}^{(C)},x_{K-1}^{(H)}), \quad\text{if}~~x_{K-1}^{(C)}=0,\!\!
\end{align}
where (i) is due to that $F(x_{K-1}, u_{K-1}, X_{K}^{(C)})$ and $\l(1+\sqrt{\dfrac{\alpha}{\lambda}}X_{K}^{(C)}\sqrt{X_{K}^{(H)}}\r)^{-2}$ are positively correlated; (ii) follows from Jensen's inequality; and (iii) is due to that $\E[(1-X_K^{(C)})| x_{K-1}, u_{K-1}]=P_{00}$, if $x_{K-1}^{(C)}=0$.
%by $x_K^{(L)}+P_{00}x_K^{(Q)}=P_{00}\l(x_K^{(L)}+x_K^{(Q)}\r)+P_{01}x_K^{(L)}$.
Last, combining the results of \eqref{eq:inf_barJK_i} and \eqref{eq:largeK_approx} and using the definition of $V(x_{K-1}^{(C)})$ in \eqref{Eq:V} lead to the desired results.

\subsection{Proof of Lemma~\ref{The:LargK-1}}\label{App:LargK-1}
We prove this lemma by solving Problem P5 for both the cases of $x_{K-1}^{(C)}=1$ and $x_{K-1}^{(C)}=0$ respectively, and then combining the results together. First, if $x_{K-1}^{(C)}=1$, the optimal solution for solving Problem P5 is same as that of zero-buffer case, given as
\begin{equation}\label{Eq:Th2X1}
u_{K-1}^{(\rm{lo})*}\!=\!x_{K-1}^{(L)}\!\!\l(1\!+\!\frac{1}{\sqrt{S(x_{K-1}^{(C)},x_{K-1}^{(H)})}}\!+\!\sqrt{\frac{\alpha}{\lambda}x_{K-1}^{(H)}}\!\r)^{-1} \text{and} \quad
u_{K-1}^{(\rm{of})*}\!=\!\sqrt{\frac{\alpha}{\lambda}x_{K-1}^{(H)}}u_{K-1}^{(\rm{lo})*}.
\end{equation}
The minimum expected user-energy consumption is
\begin{equation}\label{Eq:T2X1J}
J_{K-1}(x_{K-1})=\alpha (x_{K-1}^{(L)})^3 \l(1+\frac{1}{\sqrt{S(x_{K-1}^{(C)},x_{K-1}^{(H)})}}+\sqrt{\frac{\alpha}{\lambda}x_{K-1}^{(H)}}\r)^{-2},\quad \text{if} \quad x_{K-1}^{(C)}=1.
\end{equation}
Next, if $x_{K-1}^{(C)}=0$, we have $\widetilde{F}(x_{K-1},u_{K-1})=\l[(x_{K-1}^{(L)}\!-\!u_{K-1}^{(\rm{lo})}\!-\!u_{K-1}^{(\rm{of})})+P_{00}(x_{K-1}^{(Q)}+u_{K-1}^{(\rm{of})})\r]^3$. Substituting it into Problem P5 and solving it by the Lagrangian method, the optimal policy for Problem P5 is
\begin{equation}\label{eq:plcy_lrg_K-1_b}
\begin{split}
  u_{K-1}^{(\rm{lo})*}&\!=\!\l(x_{K-1}^{(L)}+P_{00}x_{K-1}^{(Q)}\r)\l[1+\dfrac{1}{\sqrt{S(x_{K-1}^{(C)},x_{K-1}^{(H)})}} +\sqrt{\dfrac{\alpha}{\lambda}}(1-P_{00})^{\frac{3}{2}}\sqrt{x_{K-1}^{(H)}}\r]^{-1},\\
u_{K-1}^{(\rm{of})*}&=\sqrt{\dfrac{\alpha}{\lambda}x_{K-1}^{(H)}(1-P_{00})} u_{K-1}^{(\rm{lo})*},
\end{split}
\end{equation}
%sub-optimal policy can be derived as in \eqref{eq:inf_JK-1_l*} with $V_{K-1}(x_{K-1}^{(C)})=1-P_{00}$.
%given by
%\begin{equation}\label{eq:inf_JK-1_l*}
%\begin{cases}\ell_{K-1}^*(x_{K-1})=\dfrac{\sqrt{\phi_{K-1}q_{K-1}(\bar{x}_{K-1})}\l[L_{K-1}+\l(1-\phi_{K-1}\r)Q_{K-1}\r]} {\sqrt{\beta_{K-1}}+\l[(\phi_{K-1})^{\frac{3}{2}}+\sqrt{\beta_{K-1}}\r]\sqrt{q_{K-1}(\bar{x}_{K-1})}},\\
%\tilde{\ell}_{K-1}^*(x_{K-1})=\sqrt{\dfrac{\beta_{K-1}}{\phi_{K-1}}}\ell_{K-1}^*(x_{K-1}).\end{cases}
%\end{equation}
The corresponding minimum expected user-energy consumption is
\begin{equation}\label{eq:inf_J_K-1_b}
\begin{split}
J_{K-1}(x_{K-1})&=\alpha \l(x_{K-1}^{(L)}+P_{00}x_{K-1}^{(Q)}\r)^3\\
&~~\l(1+\frac{1}{\sqrt{\hat{S}_{K-1}(x_{K-1}^{(C)},x_{K-1}^{(H)})}}+(1-P_{00})^{\frac{3}{2}}\sqrt{\frac{\alpha}{\lambda} x_{K-1}^{(H)}}\r)^{-2},\quad\text{if}\quad x_{K-1}^{(C)}=0.\!
\end{split}
\end{equation}
Last, combining the policies for the cases of $x_{K-1}^{(C)}=1$  (see \eqref{Eq:Th2X1}) and $x_{K-1}^{(C)}=0$ (see \eqref{eq:plcy_lrg_K-1_b}) gives the sub-optimal policy in slot $(K-1)$, as presented in the lemma. The expected user-energy consumption in slot $(K-1)$ can be expressed  as \eqref{eq:inf_J_K-1} by combining \eqref{Eq:T2X1J} and \eqref{eq:inf_J_K-1_b}, completing the proof.

\subsection{Proof of Lemma~\ref{lem:exist_Q_th}}\label{App:exist_Q_th}
Assume that both $G_1(Q)$ and $G_2(Q)$ are continuous functions with $Q$. First, when $Q=0$, we can obtain that $G_1(0)<G_2(0)$, since it corresponds to the case of zero buffer and ZBP leads to the optimal policy. On the other hand, when $Q=D$, we have $G_1(D)>G_2(D)$ as it refers to the case of large buffer and TLBP yields the better solution. Then, applying the intermediate-value theorem leads to the desired result.

\renewcommand\refname{Reference}
%\bibliographystyle{ieeetr}
%\bibliography{reference}

\end{document}